\documentclass[%
reprint,
nofootinbib,
 amsmath,amssymb,
showkeys,
]{revtex4-2}

\usepackage{orcidlink}
\usepackage{mathrsfs}

\usepackage{graphicx}
\usepackage{dcolumn}
\usepackage{bm}

\usepackage{amsthm}
\usepackage{xfrac}
\usepackage[mathcal]{euscript}

\usepackage{url}
\usepackage{hyperref}
\usepackage{color}
\definecolor{refcolor}{RGB}{0,0,190}
\hypersetup{
    colorlinks,
    citecolor=refcolor,
    filecolor=refcolor,
    linkcolor=refcolor,
    urlcolor=refcolor
}
\usepackage{bookmark}
\bookmarksetup{
  numbered, 
  open,
}

\newtheorem{theorem}{Theorem}

\newtheorem{definitionCustom}{Definition}
\renewcommand{\thedefinitionCustom}{\arabic{definition}}
\makeatletter
\newcommand{\setdefinitionCustomtag}[1]{
  \let\oldthedefinitionCustom\thedefinitionCustom
  \renewcommand{\thedefinitionCustom}{#1}
  \g@addto@macro\enddefinitionCustom{
    \global\let\thedefinitionCustom\oldthedefinitionCustom}
  }
\makeatother

\newtheorem{thesisCustom}{Thesis}
\renewcommand{\thethesisCustom}{\arabic{thesis}}
\makeatletter
\newcommand{\setthesisCustomtag}[1]{
  \let\oldthethesisCustom\thethesisCustom
  \renewcommand{\thethesisCustom}{#1}
  \g@addto@macro\endthesisCustom{
    \global\let\thethesisCustom\oldthethesisCustom}
  }
\makeatother

\theoremstyle{remark}

\theoremstyle{definition}

\newtheorem{observation}{Observation}
\newtheorem{problem}{Problem}

\newtheorem{condition}{Condition}
\renewcommand{\thecondition}{\arabic{condition}}
\makeatletter
\newcommand{\setconditiontag}[1]{
  \let\oldthecondition\thecondition
  \renewcommand{\thecondition}{#1}
  \g@addto@macro\endcondition{
    \global\let\thecondition\oldthecondition}
  }
\makeatother

\newtheorem{proposition}{Proposition}

\newtheorem{corollary}{Corollary}
\newtheorem{definition}{Definition}

\theoremstyle{remark}
\newtheorem{remark}{Remark}
\newtheorem{example}{Example}

\theoremstyle{remark}
\newcommand{\toyExampleName}{Toy Example}
\newtheorem{toyExample}{\toyExampleName}

\theoremstyle{definition}
\newcommand{\lessonName}{Lesson}
\newtheorem{lesson}{\lessonName}

\newcommand{\orcid}[1]{\href{https://orcid.org/#1}{\textcolor[HTML]{A6CE39}{\aiOrcid}}}
   
\def\({\left(}
\def\){\right)}
\def\[{\left[}
\def\]{\right]}

\newcommand{\tn}{\textnormal}

\newcommand{\dsfrac}[2]{\displaystyle{\frac{#1}{#2}}}

\newcommand{\boplus}{\textstyle{\bigoplus}}
\newcommand{\botimes}{\textstyle{\bigotimes}}
\newcommand{\bcup}{\textstyle{\bigcup}}

\newcommand{\MQS}{\ref{def:MQS}}
\newcommand{\HSF}{\ref{thesis:HSF}}

\newcommand{\struct}{\mc{S}}
\newcommand{\kind}{\mc{K}}

\newcommand{\hilbert}{\mathcal{H}}

\newcommand{\mc}[1]{\mathcal{#1}}
\newcommand{\ms}[1]{\mathscr{#1}}

\newcommand{\wh}[1]{\widehat{#1}}
\newcommand{\dwh}[1]{\wh{\rule{0ex}{1.3ex}\smash{\wh{\hfill{#1}\,}}}}

\newcommand{\wt}[1]{\widetilde{#1}}

\newcommand{\R}{\mathbb{R}}
\newcommand{\C}{\mathbb{C}}

\newcommand{\N}{\mathbb{N}}

\newcommand{\abs}[1]{\left|#1\right|}
\newcommand{\norm}[1]{\lVert#1\rVert}

\newcommand{\de}{\operatorname{d}}
\newcommand{\ii}{\operatorname{i}}
\newcommand{\ee}{\operatorname{e}}
\newcommand{\tr}{\operatorname{tr}}

\newcommand{\ie}{\textit{i.e.}\ }

\newcommand{\eg}{\textit{e.g.}\ }
\newcommand{\cf}{\textit{cf.}\ }
\newcommand{\etal}{\textit{et al.}}

\newcommand{\su}{\mathfrak{su}}
\newcommand{\OO}{\tn{O}}

\newcommand{\schrod}{Schr\"odinger}

\newcommand{\bra}[1]{\langle#1|}
\newcommand{\ket}[1]{|#1\rangle}
\newcommand{\kett}[1]{|\!\!|#1\rangle\!\!\rangle}
\newcommand{\braket}[2]{\langle#1|#2\rangle}
\newcommand{\ketbra}[2]{|#1\rangle\langle#2|}

\newcommand{\Herm}{\tn{Herm}}

\newcommand{\meanvalue}[2]{\langle{#1}\rangle_{#2}}

\newcommand{\x}{\mathbf{x}}
\newcommand{\xThree}{\boldsymbol{x}}
\newcommand{\pThree}{\boldsymbol{p}}

\newcommand{\q}{\mathbf{q}}
\newcommand{\p}{\mathbf{p}}

\newcommand{\n}{\mathbf{n}}

\newcommand{\intl}[1]{\int\limits_{#1}}

\def\sref #1{\S\ref{#1}}

\newcommand{\image}[3]{
\begin{figure}[!ht]
\centering
\includegraphics[width=#2\textwidth]{#1}
\caption{\label{#1}#3}
\end{figure}
}

\begin{document}

\title[3D-Space and the preferred basis cannot uniquely emerge]{3D-Space and the preferred basis\\ cannot uniquely emerge from\\ the quantum structure}

\author{Ovidiu Cristinel Stoica\ \orcidlink{0000-0002-2765-1562}}
\address{
 Dept. of Theoretical Physics, NIPNE---HH, Bucharest, Romania. \\
	Email: \href{mailto:cristi.stoica@theory.nipne.ro}{cristi.stoica@theory.nipne.ro},  \href{mailto:holotronix@gmail.com}{holotronix@gmail.com}
	}%

\begin{abstract}
Hilbert-Space Fundamentalism (HSF) states that the only fundamental structures are the quantum state vector and the Hamiltonian, and from them everything else emerge uniquely, including the 3D-space, a preferred basis, and a preferred factorization of the Hilbert space.

In this article it is shown that whenever such a structure emerges from the Hamiltonian and the state vector alone, if it is physically relevant, it is not unique.

Moreover, HSF leads to strange effects like ``passive'' travel in time and in alternative realities, realized simply by passive transformations of the Hilbert space.

The results from this article affect all theories that adhere to HSF, whether they assume branching or state vector reduction (in particular the version of Everett's Interpretation coined by Carroll and Singh ``Mad-dog Everettianism''), various proposals based on decoherence, proposals that aim to describe everything by the quantum structure alone, and proposals that spacetime emerges from a purely quantum theory of gravity.
\end{abstract}


\maketitle
\tableofcontents

\section{Introduction}
\label{s:intro}

For a closed system, Quantum Mechanics (QM) assumes a \emph{Hilbert space} $\hilbert$, a \emph{Hamiltonian operator} $\wh{H}$, and a \emph{state vector} $\ket{\psi(t)}\in\hilbert$ which depends on time according to the {\schrod} equation
\begin{equation}
\label{eq:schrod}
\ii\hbar\frac{\de}{\de t}\ket{\psi(t)}=\wh{H}\ket{\psi(t)}.
\end{equation}

\setdefinitionCustomtag{MQS}
\begin{definitionCustom}
\label{def:MQS}
In the following, we call  \emph{minimalist quantum structure (\MQS)} the triple
\begin{equation}
\label{eq:MQS}
\bigl(\hilbert,\wh{H},\ket{\psi}\bigr)
\end{equation}
where $\ket{\psi(t)}$ obeys the {\schrod} equation \eqref{eq:schrod}.
\end{definitionCustom}

Is the {\MQS} sufficient to uniquely determine all other aspects of the physical world, including the $3$D-space, a preferred basis, and a preferred factorization of the Hilbert space into factors that represent elementary particles?

According to the following Thesis, sometimes named \emph{Hilbert-space fundamentalism} (\cf \cite{Carroll2021RealityAsAVectorInHilbertSpace}), the answer is yes.

\setthesisCustomtag{HSF}
\begin{thesisCustom}[Hilbert-space Fundamentalism Thesis]
\label{thesis:HSF}
Everything about a physical system, including the $3$D-space, a preferred basis, a preferred factorization of the Hilbert space (needed to represent subsystems, {\eg} particles), emerge uniquely (or essentially uniquely) from the triple $(\hilbert,\wh{H},\ket{\psi(t)})$.
\end{thesisCustom}

It seems right that unitary symmetry requires us to interpret the wavefunction $\psi(\x)=\braket{\x}{\psi}$ on the nonrelativistic configuration space $\R^{3\n}$ as just a particular representation favored only if we pick a preferred basis $(\ket{\x})_{\x\in\R^{3\n}}$ of the Hilbert space, while the only real structure is the abstract state vector $\ket{\psi}$. Taking as fundamental the state vector $\ket{\psi}$ seems to make more sense than the wavefunction, since a preferred representation of the Hilbert space would be akin to the notion of an absolute reference frame in space. And indeed this is sometimes the stated position in the discussions about a preferred basis, a preferred factorization, or emergent space or spacetime.
This is claimed to be true in certain approaches to Quantum Gravity, formulations of Everett's Interpretation and its Many-Worlds variants (MWI) \cite{Everett1957RelativeStateFormulationOfQuantumMechanics,Everett1973TheTheoryOfTheUniversalWaveFunction,Saunders2010ManyWorldsEverettQuantumTheoryReality,Wallace2012TheEmergentMultiverseQuantumTheoryEverettInterpretation,SEP-Vaidman2018MWI}, but also in the \emph{Consistent Histories} approaches \cite{Griffiths1984ConsistentHistories,Omnes1988ConsistentHistoriesLogicalReformulationOfQuantumMechanicsIFoundations,Omnes1992ConsistentInterpretationsQuantumMechanics,Omnes1999UnderstandingQuantumMechanics,GellMannHartle1990DecoheringHistories,GellMannHartle1990QuantumMechanicsInTheLightOfQuantumCosmology}.
Presumably, decoherence \cite{Zeh1970OnTheInterpretationOfMeasurementInQuantumTheory,Zeh1996WhatIsAchievedByDecoherence,JoosZeh1985EmergenceOfClassicalPropertiesThroughInteractionWithTheEnvironment,Zurek1991DecoherenceAndTheTransitionFromQuantumToClassical,Zurek2003DecoherenceEinselectionAndTheQuantumOriginsOfTheClassical,Zurek2006DecoherenceTransitionFromQuantumToClassicalRevisited,KieferJoos1999DecoherenceConceptsAndExamples,JoosZehKieferGiuliniKupschStamatescu2003DecoherenceAndTheAppearanceOfClassicalWorldInQuantumTheory,Schlosshauer2007DecoherenceAndTheQuantumToClassicalTransition} solves the preferred basis problem and leads to the emergence of the classical world.
Therefore, claims that the preferred basis problem is solved are very common \cf Wallace \cite{Wallace2010DecoherenceAndOntology,Wallace2012TheEmergentMultiverseQuantumTheoryEverettInterpretation}, Tegmark \cite{Tegmark1998InterpretationQMManyWorldsOrManyWords}, Brown and Wallace \cite{BrownWallace2005BohmVsEverett}, Zurek \cite{Zurek2003DecoherenceEinselectionAndTheQuantumOriginsOfTheClassical}, Schlosshauer \cite{Schlosshauer2006MinimalNoCollapseQuantumMechanics,Schlosshauer2007DecoherenceAndTheQuantumToClassicalTransition}, Saunders \cite{Saunders2010ManyWorldsAnIntroduction,Saunders2010ManyWorldsEverettQuantumTheoryReality} \textit{etc}.

In the case of MWI, this claim was criticized for circularity \cite{Kent1990AgainstManyWorldsInterpretations}.
Other authors stated clearly that the $3$D-space and a factorization are necessary prerequisites of the theory \cite{Vaidman2018OntologyOfTheWavefunctionAndMWI,Wallace2012TheEmergentMultiverseQuantumTheoryEverettInterpretation}.

But Carroll and Singh claimed to prove that the Hamiltonian, moreover, its spectrum, is sufficient to determine an essentially unique space structure (\cite{CarrollSingh2019MadDogEverettianism,Carroll2021RealityAsAVectorInHilbertSpace}, p. 99)
\begin{quote}
a generic Hamiltonian will not be local with respect to any decomposition, and 
for the special Hamiltonians that can be written in a local form, the
decomposition in which that works is essentially unique.
\end{quote}

In \cite{CarrollSingh2019MadDogEverettianism}, p. 95, they wrote about the {\MQS} that
\begin{quote}
Everything else--including space and fields propagating on it--is emergent from these minimal elements.
\end{quote}

Carroll and Singh based their reconstruction of space on the results obtained by Cotler {\etal} \cite{CotlerEtAl2019LocalityFromSpectrum} regarding the uniqueness of factorization of the Hilbert space, so that the interaction encoded in the Hamiltonian is ``local'' in a certain sense. The result obtained by Cotler {\etal} states in fact that such a factorization is ``almost always'' unique (\cite{CotlerEtAl2019LocalityFromSpectrum}, p. 1267).
We will see that these emerging structures, if physically relevant, are not unique.

In this article we will give general proofs that, whenever the Hamiltonian leads to a tensor product decomposition of the Hilbert space, a $3$D-space structure, or a preferred generalized basis, it leads to many physically distinct structures of the exact same type, where by ``physically distinct'' it is understood that they exhibit physical differences.
In the discussion of recovering space, for simplicity, we will focus on $3$D physical spaces, but the proofs in this article apply to any number of dimensions.

\definecolor{greenPsi}{rgb}{0.0, 0.375, 0.0}
\definecolor{blueStruct}{rgb}{0.0, 0.0, 1.0}
\definecolor{redStruct}{rgb}{1.0, 0.0, 0.0}
\image{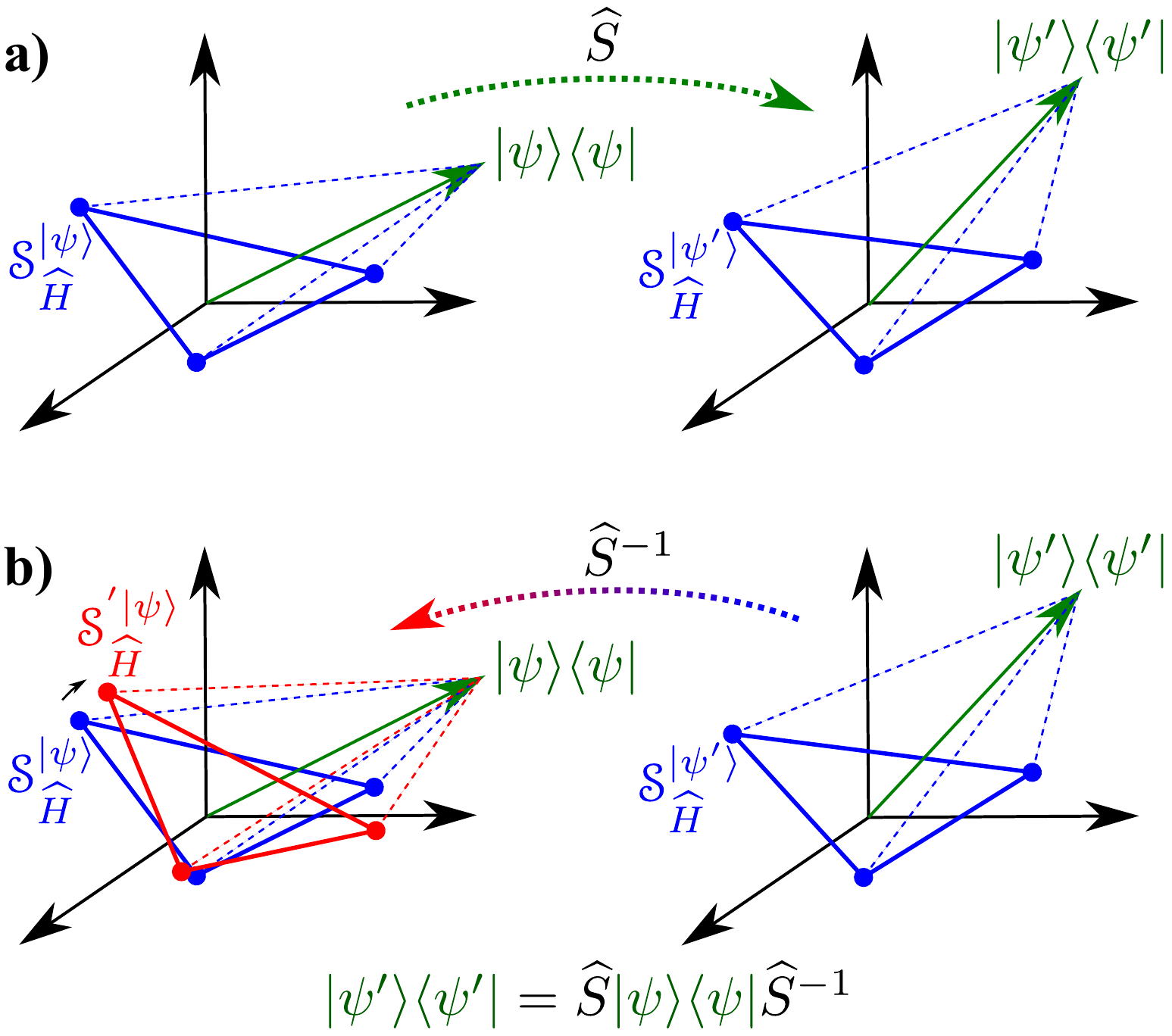}{0.45}{Schematic representation of the idea of the proof in the vector space of Hermitian operators (see Theorem \ref{thm:nogo}).
\\
{\bf a)} The candidate preferred structures $\textcolor{blueStruct}{\struct_{\wh{H}}^{\ket{\psi}}}$ and $\textcolor{blueStruct}{\struct_{\wh{H}}^{\ket{\psi'}}}$ are symbolized as solid blue triangles. The vertices represent the preferred observables of the structure, and the edges relations between the observables.
The dashed blue lines symbolize their relations with the corresponding states, represented by $\ket{\psi}\bra{\psi}$ and $\ket{\psi'}\bra{\psi'}=\wh{S}\ket{\psi}\bra{\psi}\wh{S}^{-1}$, where $[\wh{S},\wh{H}]=0$.
For the candidate $3$D-space and other candidate preferred structures expected to emerge, for any state vector $\ket{\psi}$ there are state vectors $\ket{\psi'}\neq\ket{\psi}$ for which these relations are different. Without this, the structure is not physically relevant. An obvious example is obtained by taking $\wh{S}$ to be a time evolution operator, because we expect that $\ket{\psi}$ changes in time with respect to space, the preferred bases, and the preferred factorization.
In addition to the time evolution ones, other choices of $\wh{S}$ give different families of unitary transformations.
\\
{\bf b)} Since the structure $\textcolor{blueStruct}{\struct_{\wh{H}}^{\ket{\psi}}}$ is invariant to any unitary transformation $\wh{S}$ that commutes with $\wh{H}$, $\wh{S}^{-1}$  transforms $\textcolor{blueStruct}{\struct_{\wh{H}}^{\wh{S}\ket{\psi}}}$ into another structure for $\ket{\psi}$, $\textcolor{redStruct}{\struct_{\wh{H}}^{'\ket{\psi}}}=\wh{S}^{-1}\[\textcolor{blueStruct}{\struct_{\wh{H}}^{\wh{S}\ket{\psi}}}\]$, which is of the same kind as $\textcolor{blueStruct}{\struct_{\wh{H}}^{\ket{\psi}}}$.
Since the relations between $\textcolor{redStruct}{\struct_{\wh{H}}^{'\ket{\psi}}}$ and $\ket{\psi}$ are different from those between $\textcolor{blueStruct}{\struct_{\wh{H}}^{\ket{\psi}}}$ and $\ket{\psi}$, we conclude that $\textcolor{blueStruct}{\struct_{\wh{H}}^{\ket{\psi}}}$ is not unique. Theorem \ref{thm:nogo} is more general, since it also covers uniqueness up to a physical equivalence.}

In \sref{s:simple-examples} we illustrate the main ideas by providing simple examples and learning from them.
In \sref{s:examples-NRQM-space}, we show that there are infinitely many physically distinct choices of the $3$D-space in nonrelativistic QM. These examples illustrate the main idea of the proof, which will be given in full generality in Sec. \sref{s:proof-TS}. 

The kinds of candidate preferred structures are defined in full generality in \sref{s:proof-TS:kinds}.
The conditions of uniqueness (up to non-physical differences) and physical relevance (the ability to distinguish physically distinct states) are the object of \sref{s:proof-TS:cond}.
The main theorem, showing that if a candidate preferred structure is physically relevant there are more physically distinct such structures, is proved in \sref{s:proof-TS:thm}.
The idea of the proof is illustrated in Fig. \ref{sketch-proof.pdf}. 

In \sref{s:abundance} we show that non-uniqueness is very rich, and we discuss an interesting particular type of Hamiltonian.

In Sec. \sref{s:proof-TS-applications} we apply the main theorem from Sec. \sref{s:proof-TS} to prove the non-uniqueness of generalized ``preferred'' bases in \sref{s:proof-PBS}, of factorizations into subsystems in \sref{s:proof-TPS}, of $3$D-space structures as proposed by Carroll \& Singh and by Giddings in \sref{s:proof-TPS-space}, of $3$D-space structures in Quantum Gravity and in general in \sref{s:proof-emergent-space}, of generalized bases based on coherent states in \sref{s:proof-coherent_states}, of environmental decoherence in \sref{s:proof-PBS-sub}, and of emergent quasi-classicality, exemplified with Carroll \& Singh's \emph{quantum mereology} in \sref{s:proof-classicality}.

In Sec. \sref{s:passive_travel} we show that assuming the {\MQS} to be the only fundamental structure has strange consequences: the state vector representing the present state equally represents all the past and future states, as well as alternative realities.
In Sec. \sref{s:avoid-problem} we discuss the possible options with which the results from this article leave us, connecting them with the various interpretations of QM.

\section{Main ideas of the proof}
\label{s:main-ideas}

In this Section we will illustrate the main ideas of the proof, given in Sec. \sref{s:proof-TS}, by examples.

In \sref{s:simple-examples} we will give very simple examples and extract lessons from them, to understand the main ideas.

Then, in \sref{s:examples-NRQM-space}, we will illustrate the main idea with an explicit proof for the case of nonrelativistic Quantum Mechanics (NRQM).

\subsection{Learning from simple examples}
\label{s:simple-examples}

\begin{toyExample}
\label{toyExample:finite-dim-real}
Consider a $3$D real vector space, endowed with the usual Euclidean scalar product.
We assume as \emph{given structure} a unit vector $\vec{u}$, and we need to recover from it a \emph{structure of interest}, for example a unit vector $\vec{v}$ required to be perpendicular to $\vec{u}$.

If we have a solution $\vec{v}_0\perp \vec{u}$, all other solutions can be obtained by \emph{active} rotations of $\vec{v}_0$ around $\vec{u}$. If $\vec{v}_1:=\vec{u}\times \vec{v}_0$, then the general solution is $\vec{v}(\theta)=\cos\theta \vec{v}_0+\sin\theta \vec{v}_1$.

Therefore, there are infinitely many solutions, parametrized by a continuous parameter $\theta\in[0,2\pi)$.

For a general dimension $n$, the general solution depends on $(n-1)(n-2)/2$ continuous parameters, since the transformations that leave invariant the structure of an Euclidean vector space and a unit vector form a group isomorphic to the orthogonal group $\OO(n-1)$.
\end{toyExample}

\begin{remark}
\label{rem:toyExample:finite-dim-real:physical}
A physical interpretation of {\toyExampleName} \ref{toyExample:finite-dim-real} is that we know the direction of the electric field at a point $\x\in\R^3$, and we want to find out the direction of the magnetic field at $\x$.
\end{remark}

{\toyExampleName} \ref{toyExample:finite-dim-real} may seem too limited to apply to general structures that can be recovered from the {\MQS}, but the difference is only ``quantitative''.
To see this, we gradually generalize {\toyExampleName} \ref{toyExample:finite-dim-real} in the following ways.

\begin{toyExample}
\label{toyExample:finite-dim-real-line}
Instead of a unit vector $\vec{u}\in\R^n$, the given structure can be a line, and we want to find out a unit vector perpendicular to it. Then, the result will be the same as in {\toyExampleName} \ref{toyExample:finite-dim-real}.
But if we are given a line and we want to find a line perpendicular to it, since two opposite unit vectors determine the same line, in the $3$D Euclidean case the parameter $\theta$ should be restricted to the interval $[0,\pi)$ to avoid duplicate solutions.
\end{toyExample}

\begin{toyExample}
\label{toyExample:finite-dim-real-angle}
Suppose we are given a line $a$ in the Euclidean vector space $\R^3$ and the structure of interest is another line $b$ that makes an angle $\alpha$ with $a$. Since in a vector space a line can be determined by a unit vector $\vec{u}$ by $a=\{c\vec{u}|c\in\R\}$, the line contains the origin.

If we know a solution $b_0$, we can obtain all the other solutions by rotating $b_0$ around $a$. If $\alpha\in(0,\pi/2)$, the different solutions are parametrized by the rotation angle $\theta\in[0,2\pi)$, similar to {\toyExampleName} \ref{toyExample:finite-dim-real}.

If $\alpha=\pi/2$, the different solutions are parametrized by the rotation angle $\theta\in[0,\pi)$, as in {\toyExampleName} \ref{toyExample:finite-dim-real-line}.

If $\alpha=0$, then there is a unique solution, $b=a$.
\end{toyExample}

\begin{toyExample}
\label{toyExample:finite-dim-real-subspaces}
Suppose that the given structure is a subset of vector subspaces of the Euclidean vector space $\R^n$, and the structure of interest a vector $\vec{v}\in\R^n$. The condition defining $\vec{v}$ should then be replaced by various conditions expressing the relations between $\vec{v}$ and these subspaces.
For example, $\vec{v}$ may be required to be perpendicular to some of them, to be included in one of them, or, in general, to make various angles with each of them.

More generally, we can replace both the given structure and the structure of interest with multiple vectors, lines, subspaces of $\R^n$, circles, spheres, ellipsoids, hyperboloids and various other submanifolds of $\R^n$, required to be in certain relations with one another, by specific conditions involving angles, distances, or, in general, scalar products.
\end{toyExample}

\begin{toyExample}[General {\MQS}]
\label{toyExample:hilbert}
We can generalize {\toyExampleName} \ref{toyExample:finite-dim-real-subspaces} by replacing the Euclidean vector space $\R^n$ with a complex Hilbert space $\hilbert$ of any dimension.
The given structure can consist of the subspace spanned by a vector $\ket{\psi}\in\hilbert$, and a set of subspaces that form an orthogonal decomposition of $\hilbert$, each of them labeled by a different real number.

This set of labeled subspaces is equivalent to a self-adjoint operator $\wh{H}$, the subspaces being its eigenspaces, and the labels the corresponding eigenvalues.
In other words, the given structure can be a {\MQS}, so this {\toyExampleName} covers the general situation from Thesis {\HSF}.

The structure to be recovered can consist, in its turn, of vectors from $\hilbert$, of subspaces of $\hilbert$, of self-adjoint operators on $\hilbert$, all of them required to be in certain relations both with the elements of the given structure and with one another.
The most general structures of interest will be discussed in Sec. \sref{s:proof-TS}.
\end{toyExample}

From these examples we learn several lessons.

\begin{lesson}
\label{lesson:structures}
Consider the problem of showing that a certain \emph{structure of interest} emerges out of a \emph{given structure}.
In the case of Thesis {\HSF}, both these structures consist of objects defined on the Hilbert space $\hilbert$.
The given structure is the {\MQS} $(\hilbert,\wh{H},\ket{\psi})$.
The structure of interest is similarly defined in terms of vectors, operators, \textit{etc}.
The precise conditions to be satisfied by a structure of interest will be said to specify its \emph{kind} (see \sref{s:proof-TS:kinds}).
\end{lesson}

\begin{lesson}
\label{lesson:tensors}
The most general objects constituting a structure in a vector space $\hilbert$ are the \emph{tensors} on $\hilbert$.
In particular, numbers (scalars), vectors, and linear operators are tensors. Closed subspaces can be defined as eigenspaces of linear operators. Submanifolds of $\hilbert$ of various degrees can be defined by tensors of various degrees, just like in Euclidean geometry quadric hypersurfaces are defined by bilinear forms.
\end{lesson}

In {\toyExampleName} \ref{toyExample:finite-dim-real} the structure of interest was a vector required to be perpendicular on the given vector.
In {\toyExampleName}
\ref{toyExample:finite-dim-real-subspaces} the relations were expressed using angles or distances.
These are invariants in Euclidean geometry, based on the scalar product.
For the complex case from {\toyExampleName}~\ref{toyExample:hilbert}, we should use the Hermitian scalar product of the $\hilbert$ space.

\begin{lesson}
\label{lesson:relations}
The objects of the structure of interest are required to \emph{satisfy specific conditions}, expressed as relations with one another and with the objects of the structure of interest.
\end{lesson}

\begin{lesson}
\label{lesson:invariance-all}
If we transform simultaneously all the objects involved while keeping the structure of the Hilbert space, we need to use unitary transformations.
Under such transformations, the relations between the objects remain unchanged.
Therefore, these relations can be expressed, using the scalar products, in terms of \emph{invariants} \cite{Weyl1946ClassicalGroupsInvariantsAndRepresentations,Rota2001WhatIsInvariantTheoryReally}.
\end{lesson}

Lessons \ref{lesson:structures}, \ref{lesson:tensors}, \ref{lesson:relations}, \ref{lesson:invariance-all} will be taken into account in \sref{s:proof-TS:kinds}, where the general definition of the kind of a structure will be given in terms of tensors and conditions, which are invariant equations or inequations.

\begin{remark}
\label{rem:active-not-passive}
In {\toyExampleName} \ref{toyExample:finite-dim-real} we have seen that by knowing the given vector $\vec{u}$ and a vector $\vec{v}_0$ satisfying the desired condition, all other solutions could be found from vector $\vec{v}_0$ by \emph{active} rotations that preserve the vector $\vec{u}$.
The rotations have to be \emph{active, not passive}, because a passive rotation would lead to the same solution expressed in a different basis.
\end{remark}

\begin{lesson}
\label{lesson:invariance-fix}
If we know a structure of interest satisfying the defining conditions, we can find other similar structures by applying \emph{active} symmetry transformations that preserve the given structure.
In the case of a Hilbert space, the symmetry transformations are unitary transformations.
If the structure of interest is supposed to be determined by the Hamiltonian $\wh{H}$, the transformations that can be used have to commute with $\wh{H}$.
If the structure of interest is supposed to be determined by $\wh{H}$ and $\ket{\psi}$, the transformations that can be used have to commute with $\wh{H}$ and leave $\ket{\psi}$ unchanged.
\end{lesson}

\begin{lesson}
\label{lesson:unique}
It is possible that a given structure determines a unique structure of a particular kind.

In {\toyExampleName} \ref{toyExample:finite-dim-real-angle}, if $\alpha$ is a multiple of $\pi$, the structure of interest is unique, because in this case it has to be identical with the given structure.

Another example: if the given structure consists of two distinct concurrent lines, and the structure of interest is their intersection point, then this can be unique.

It is also possible that the structure of interest is uniquely recovered from the {\MQS}. For example, the eigenspaces and the eigenvalues of $\wh{H}$ are uniquely determined by $\wh{H}$. Also, any operator defined as a polynomial of $\wh{H}$ is unique (Remark \ref{rem:unique}).
\end{lesson}

\begin{remark}
\label{rem:unique-gauge}
We should take into account the possibility that even if the solution is not unique, distinct solutions may be \emph{physically equivalent}, \ie they can be transformed into one another by non-physical symmetry transformations.
For example, if we want to determine a vector $\ket{\phi}$, we should expect that $\ee^{i\theta}\ket{\phi}$, $\theta\in\R$, is also a solution satisfying the same conditions that define the structure of interest.
Maybe some unitary transformations that preserve both $\wh{H}$ and $\ket{\psi}$ and transform a solution into another one are phase transformations or gauge transformations.
If all solutions can be related by structure related by such ``unphysical transformations'', we will say that the solution is \emph{essentially unique}, even though it is not strictly unique.
We will discuss this in \sref{s:proof-TS:cond} and take it into account when proving the main result.
\end{remark}

\begin{remark}
\label{rem:too-unique}
If all solutions can be related by unitary transformations that leave $\wh{H}$ (and possibly $\ket{\psi}$) unchanged (as in Remark \ref{rem:unique-gauge}, one may think that they should be considered physically equivalent.
Then, one may say that the solution is essentially unique.
For example, all transformations that leave the vector $\vec{u}$ from {\toyExampleName} \ref{toyExample:finite-dim-real} unchanged can be declared ``unphysical'', so the solution can then be said to be ``essentially unique''.

However, there are structures that are in different relations with different state vectors. For example, an invariant characterizing the relation between an operator $\wh{A}$ and a state vector $\ket{\psi}$ is the mean value $\bra{\psi}\wh{A}\ket{\psi}$. But this can be different for different state vectors $\ket{\psi}\neq\ket{\psi'}$, since in general $\bra{\psi}\wh{A}\ket{\psi}\neq\bra{\psi'}\wh{A}\ket{\psi'}$.
We expect this to be true for example when $\wh{A}$ is a position operator, which means that space itself can be in different relations with different state vectors.
Similarly for a preferred basis, when $\wh{A}$ represents projectors.
And the same is true for a tensor product structure $\hilbert_{A}\otimes\hilbert_{B}$, since the reduced density operator of the system $A$ can be physically different for different states, $\tr_{B}\ket{\psi}\bra{\psi}\neq\tr_{B}\ket{\psi'}\bra{\psi'}$.
In all these cases, the different relations can be observed physically.

Two physically distinct states may be distinguished experimentally, even if they are not orthogonal, because the probability that the outcomes of the measurements are distinct is nonvanishing. They can be distinguished with certainty provided that they are orthogonal or that we can repeat the measurements sufficiently many times. Similarly, physically distinct observables measured on the same state can result in different outcomes.
\end{remark}

From Remark \ref{rem:too-unique} we learn that

\begin{lesson}
\label{lesson:unique-loose}
Some structures of interest, for example the $3$D-space, the tensor product structure, and a preferred basis, can be in different relations with different state vectors, and the differences may be physical.
Such structures will be called \emph{physically relevant}.
We will see in Theorem \ref{thm:nogo} that, regardless how loosely we define ``essential uniqueness'', if the structure is physically relevant, then there will be more physically distinct solutions (Definition \ref{def:equiv}) that can be distinguished by their physical observables.
\end{lesson}

\begin{remark}
\label{rem:preferred}
But what should we understand by a ``preferred structure''?
If a unique, or essentially unique structure of a particular kind emerges, it is said to be \emph{preferred}. Until it is decided whether it is preferred or not, it will be named \emph{candidate preferred structure}.

In other words, our main result, Theorem \ref{thm:nogo}, does not deny the existence of a preferred structure of a particular kind, nor does it claim it exists. It only shows that, \emph{if} such a structure exists \emph{and} it is physically relevant ({\lessonName} \ref{lesson:unique-loose}), \emph{then} it is not unique.

The notion of a ``preferred'' structure of a particular kind can have various meanings, as given by various authors, but we will see that there is a way to treat them in a unified way, independent of the definitions used by various authors.
Broadly speaking, when we are interested in a structure that is not directly specified in addition to the {\MQS} (which only specifies $\hilbert$, $\wh{H}$, and $\ket{\psi}$), such a structure should be characterized by some properties or conditions ({\lessonName}
\ref{lesson:relations}). These conditions depend on the situation and on the approach, but in each case they restrict the possible structures of interest to a particular ``kind'' of structures. Then, if the structures of the same kind have a unique correspondent in reality, we expect that there is a unique structure of that kind, up to physically irrelevant differences. 
So, briefly, when our theory is compatible with more structures of a particular kind, but in reality we know that physically there is only one such structure, we expect that such a structure stands out in a preferred way.
This is the rationale behind Thesis {\HSF}.

If Thesis {\HSF} is correct, this should also apply to Quantum Mechanics as usually formulated, but also to any quantum theory, for example one aiming to include Quantum Gravity which does not assume as given these structures, and claims that they can be recovered uniquely.
If one claims that they can be recovered just from the {\MQS}, they should emerge as preferred structures of a particular kind, in the sense that they have to satisfy certain mathematical conditions, which may depend on the particular theory and the particular proposal to recover them as emergent structures.
If they are supposed to correspond to unique structures from the real world, they should be recovered uniquely, or at least up to some physically irrelevant differences.

Even in the standard formulation of Quantum Mechanics, where all the physically relevant observables are represented by the position and momentum operators and all operators that can be built out of them, there may be structures that are not given \emph{a priori}, but are thought to emerge, for example, the ``preferred basis'' for decoherence.
There are various different proposals to characterize the preferred basis. In general they involve the positions in the classical configuration or phase space of the underlying classical theory that was quantized.
Various proposals give different definitions of the preferred basis, which translate into different kinds of structures and constraining conditions.

All these candidate preferred structures will be treated in a unified way in this article.
While each of these structures is defined in a particular way by the researchers that propose ways of emergence for them, the main theorem from this article is formulated in the most general way possible, so that it applies to any characterization of a preferred structure. The results are not limited to the structures mentioned in the Thesis \ref{thesis:HSF}, being formulated in full generality.
Even if this makes them less intuitive, a general proof may help researchers realize that it just does not matter how they redefine $3$D-space or other structures, or what procedure to recover them they use, because the uniqueness is ruled out for any of them that is physically relevant ({\lessonName} \ref{lesson:unique-loose}).
We will say more about this in Sec. \sref{s:proof-TS}.
\end{remark}

\subsection{Example: Non-uniqueness of space in nonrelativistic Quantum Mechanics}
\label{s:examples-NRQM-space}

Let us consider the following Hamiltonian operator for $\n$ particles in NRQM,
\begin{equation}
\label{eq:schrod_hamiltonian_NRQM}
\wh{H} = \sum_{j=1}^{\n}\frac{1}{2m_j}\wh{\pThree}_{j}^2
+ \mathop{\sum_{j=1}^{\n}\sum_{k=1}^{\n}}_{j\neq k}\wh{V}_{j,k}\(\abs{\wh{\xThree}_j-\wh{\xThree}_k}\),
\end{equation}
where 
$m_j$ is the mass of the particle $j$
and $\wh{\pThree}_{j}^2=\wh{\pThree}_{j}\cdot\wh{\pThree}_{j}$.

For each particle $j$, the operators $\wh{\xThree}_j$ and $\wh{\pThree}_j$ are $3$-vector operators, in the sense that, when expressed in coordinates $(x_{j1},x_{j2},x_{j3})$ of the $3$D-space of the particle $j$, each of them has three operators as components,
\begin{equation}
\label{eq:xpThreeBasis}
\begin{aligned}
\wh{\xThree}_j&=\(\wh{x}_{j1},\wh{x}_{j2},\wh{x}_{j3}\),\\
\wh{\pThree}_{j}&=\(\wh{p}_{j1},\wh{p}_{j2},\wh{p}_{j3}\)\\
&=\(-\ii\hbar\dsfrac{\partial}{\partial x_{j1}},-\ii\hbar\dsfrac{\partial}{\partial x_{j2}},-\ii\hbar\dsfrac{\partial}{\partial x_{j3}}\)\\
&=-\ii\hbar\nabla_{\xThree_j}.\\
\end{aligned}
\end{equation}

In the position representation, the state vector takes the form of a wavefunction defined as 
\begin{equation}
\label{eq:psi_wave_function_x}
\psi(\xThree_1,\ldots,\xThree_{\n},t) := \braket{\xThree_1,\ldots,\xThree_{\n}}{\psi(t)},
\end{equation}
which belongs to the Hilbert space of complex square-integrable functions $L^2\(\R^{3\n}\)$, where $(\xThree_1,\ldots,\xThree_{\n})\in\R^{3\n}$ represents a point in the configuration space.

\begin{remark}
\label{rem:3dspace:nrqm:specs}
In NRQM, the $3$D-space structure is specified by the position operators. 
Following Remark \ref{rem:preferred}, we need to characterize a $3$D-space structure by providing a mathematical structure on the Hilbert space, along with certain conditions to be satisfied by this structure. We can use the following conditions for this:
\begin{enumerate}
	\item There is a complete set of commuting operators $(\wh{x}_{ja})_{a\in\{1,2,3\}, j\in\{1,\ldots,\n\}}$.
	\item The spectrum of each operator $\wh{x}_{ja}$ is $\R$.
	\item The Hamiltonian has the form \eqref{eq:schrod_hamiltonian_NRQM}, where each $\wh{p}_{ja}$ is conjugate of $\wh{x}_{ja}$, so it satisfies $\wh{p}_{ja}=-\ii\hbar\dsfrac{\partial}{\partial x_{ja}}$, where $x_{ja}$ belongs to the spectrum of $\wh{x}_{ja}$.
\end{enumerate}

The position operators specify the positions in the configuration space $\R^{3\n}$, not in the $3$D-space.
But if the interaction part of the potential depends on the distance between the particles in a way that makes it clear when the distance between two particles vanishes, the subspaces of $\R^{3\n}$ corresponding to the position space of each of the $\n$ particles can be mapped one into the other, allowing the recovery of the $3$D-space, see \eg \cite{Albert1996ElementaryQuantumMetaphysics}.
For a set of operators that gives a more direct specification of the $3$D-space see the discussion from \sref{s:proof-emergent-space}.

We will show that there are infinitely many ways to choose a set of operators that have the same mathematical properties as the position and momentum operators, and so that the Hamiltonian has the same form when expressed in terms of these operators.
\end{remark}

Let us consider a unitary transformation $\wh{S}$ of the Hilbert space $\hilbert$.
The transformations by $\wh{S}$ of the operators $\wh{\xThree}_j$ and $\wh{\pThree}_{j}$, for all $j\in\{1,\ldots,\n\}$, are
\begin{equation}
\label{eq:position_operators_transformed}
\begin{cases}
\wh{\wt{\xThree}}_j &:= \wh{S}\wh{\xThree}_j\wh{S}^{-1} \\
\wh{\wt{\pThree}}_j &:= \wh{S}\wh{\pThree}_j\wh{S}^{-1} =
-\ii\hbar\nabla_{\wt{\xThree}_j}, \\
\end{cases}
\end{equation}
where $\wt{\xThree}_j=(\wt{x}_{j1},\wt{x}_{j2},\wt{x}_{j3})$ and each $\wt{x}_{ja}$ belongs to the spectrum of $\wh{\wt{x}}_{ja}$.

Since the spectrum of each of the operators $\wh{x}_{ja}$ and $\wh{p}_{ja}$ is $\R$ and the transformation $\wh{S}$ is unitary, the spectrum of each of the operators $\wh{\wt{x}}_{ja}$ and $\wh{\wt{p}}_{ja}$ is $\R$ as well.

The operator $\wh{S}$ transforms the operators as in \eqref{eq:position_operators_transformed}, and we can represent the state vector $\ket{\psi}$ as a wavefunction
\begin{equation}
\label{eq:psi_wave_function_xi}
\wt{\psi}(\wt{\xThree}_1,\ldots,\wt{\xThree}_{\n},t) = \braket{\wt{\xThree}_1,\ldots,\wt{\xThree}_{\n}}{\psi(t)},
\end{equation}
where
\begin{equation}
\label{eq:position_basis_transformed}
\ket{\wt{\xThree}_1,\ldots,\wt{\xThree}_{\n}} = \wh{S}\ket{\xThree_1,\ldots,\xThree_{\n}}.
\end{equation}

The wavefunction \eqref{eq:psi_wave_function_xi} is also a square-integrable complex function from a space $L^2\(\R^{3\n}\)$, but in general this space is not the same as the one of the wavefunctions of positions.
In general, the parameters $(\wt{\xThree}_j)_j$ are not coordinate transformations of the parameters $(\xThree_j)_j$.
For example, the Fourier transform can be used to move to a momentum representation. But, except for a zero-measure set, most unitary transformations will not even be transformations of the coordinates in the phase space.

The dependence of the potential term $\wh{V}_{j,k}$ on the $3$D distance in the position representation suggests the following.
If there is no interaction, the Hamiltonian reduces to the kinetic term
\begin{equation}
\label{eq:hamiltonian_kinetic}
\wh{T} = \sum_{j=1}^{\n}\frac{1}{2m_j}\sum_{a=1}^{3}\wh{p}_{ja}^2,
\end{equation}
and, unless each particle has a different mass $m_j$, we cannot even recover the number of dimensions, since the only contribution of the three dimensions appears in the expression of the potentials, see \eg \cite{Albert1996ElementaryQuantumMetaphysics}.
But if the potential involves all particles, it is possible to recover the $3$D-space, provided that we know the position configuration space. In this case, we only need to focus on recovering the position configuration space.

The possibility to recover the $3$D-space from the potentials is the reason why, for this example, we also demand that the Hamiltonian has the form \eqref{eq:schrod_hamiltonian_NRQM}.
And it does:

\begin{proposition}
\label{thm:nrqm_hamiltonian-form}
After a unitary transformation $\wh{S}$, the Hamiltonian has the same form as in \eqref{eq:schrod_hamiltonian_NRQM}, but expressed in terms of the operators $(\wh{\wt{\xThree}}_j)_j$ and $(\wh{\wt{\pThree}}_j)_j$ from eq. \eqref{eq:position_operators_transformed}, instead of $(\wh{\xThree}_j)_j$ and $(\wh{\pThree}_j)_j$.
\end{proposition}
\begin{proof}
The kinetic term $\wh{T}$ in \eqref{eq:hamiltonian_kinetic} is a function $T(\wh{\pThree}_1,\ldots,\wh{\pThree}_{\n})$ of the momentum operators $\wh{\pThree}_j$.
Since they all commute, the transformation $\wh{S}\wh{T}\wh{S}^{-1}$ of $\wh{T}$ is the same function $T(\wh{\wt{\pThree}}_1,\ldots,\wh{\wt{\pThree}}_{\n})$, but of the operators $\wh{\wt{\pThree}}_j$.

Similarly, the potential term of the Hamiltonian is a function $\wh{V}(\wh{\xThree}_1,\ldots,\wh{\xThree}_{\n})$ of the position operators, and since all the position operators commute with one another, 
the transformation $\wh{S}\wh{V}\wh{S}^{-1}$ of $\wh{V}$ is the same function $V(\wh{\wt{\xThree}}_1,\ldots,\wh{\wt{\xThree}}_{\n})$, but of the operators $\wh{\wt{\xThree}}_j$.

Therefore, the form of the Hamiltonian remains the same as in eq. \eqref{eq:schrod_hamiltonian_NRQM}, but expressed in terms of the operators $(\wh{\wt{\xThree}}_j)_j$ and $(\wh{\wt{\pThree}}_j)_j$ from eq. \eqref{eq:position_operators_transformed}.
\end{proof}

This suggests that not all reparametrizations defined by unitary transformations \eqref{eq:position_operators_transformed} recover the original $3$D-space, even when $\wh{S}$ commutes with $\wh{H}$. Is it then possible to uniquely recover the $3$D-space from the {\MQS} alone?

\begin{theorem}
\label{thm:NRQM_nogo}
Any procedure to recover the $3$D-space from the NRQM Hamiltonian, and even from the Hamiltonian and the state vector, leads to infinitely many physically distinct solutions.
\end{theorem}
\begin{proof}
We prove the two claims one at a time.

\textbf{I)} \emph{Recovering the $3$D-space from the Hamiltonian only.}

Suppose we found a set of candidate preferred position operators, in which the wavefunction has the form $\psi(\xThree_1,\ldots,\xThree_{\n},t)$, and the Hamiltonian is expressed as in eq. \eqref{eq:schrod_hamiltonian_NRQM}. 
Let us see if we can find other parametrizations that look like the $3$D-space at the time $t_0$.

We take as transformation in \eqref{eq:position_operators_transformed} a unitary transformations $\wh{S}$ commuting with $\wh{H}$, so that $\ket{\psi(t_0)}$ and
\begin{equation}
\label{eq:S_from_S_k_s_inv}
\ket{\psi'(t_0)}:=\wh{S}^{-1}\ket{\psi(t_0)}.
\end{equation}
are not related  by a space isometry or a gauge transformation.
The transformation $\wh{S}$ is \emph{active} (see Remark \ref{rem:active-not-passive} and {\lessonName} \ref{lesson:invariance-fix}).
Since infinitely many linearly independent self-adjoint operators commute with $\wh{H}$, they generate infinitely many families of such unitary transformations. A particular example is
\begin{equation}
\label{eq:S_time_evolution}
\wh{S}=\wh{U}_{t,t_0}:=\ee^{-\frac{\ii}{\hbar}\wh{H}(t-t_0)}
\end{equation}
for some $t\neq t_0$.
Then, $\ket{\psi(t_0)}$ and $\ket{\psi'(t_0)}$ represent distinct states. Since in general there are physical changes in time, there are infinitely many examples of state vectors $\ket{\psi'(t_0)}$ that are physically distinct from $\ket{\psi(t_0)}$.

We apply $\wh{S}$ as a transformation at $t_0$, as in eq. \eqref{eq:position_operators_transformed}.
In the new parametrization $\wt{\xThree}_j$, $\ket{\psi(t_0)}$ has the form \eqref{eq:psi_wave_function_xi}.

If our new $3$D-space structure cannot be distinguished physically from the old one, we should expect that, in the new $3$D-space structure, $\ket{\psi(t_0)}$ looks the same as in the old one.
But when we check if this is true, we find out that it is not the case:
\begin{equation}
\label{eq:psi_wave_function_xi_thm}
\begin{aligned}
\wt{\psi}(\wt{\xThree}_1,\ldots,\wt{\xThree}_{\n},t_0) 
&\stackrel{\eqref{eq:psi_wave_function_xi}}{=} \braket{\wt{\xThree}_1,\ldots,\wt{\xThree}_{\n}}{\psi(t_0)} \\
&\stackrel{\eqref{eq:position_basis_transformed}}{=} \(\bra{\xThree_1,\ldots,\xThree_{\n}}\wh{S}^\dagger\)\ket{\psi(t_0)} \\
&\ =\ \bra{\xThree_1,\ldots,\xThree_{\n}}\(\wh{S}^{-1}\ket{\psi(t_0)}\) \\
&\stackrel{\eqref{eq:S_from_S_k_s_inv}}{=} \braket{\xThree_1,\ldots,\xThree_{\n}}{\psi'(t_0)} \\
&\ =\ \psi'(\xThree_1,\ldots,\xThree_{\n},t_0). \\
\end{aligned}
\end{equation}

Then, in the configuration space of positions obtained by using the unitary transformation $\wh{S}$, the wavefunction is identical to $\psi'(\xThree_1,\ldots,\xThree_{\n},t_0)$, which represents a physically distinct state from $\psi(\xThree_1,\ldots,\xThree_{\n},t_0)$. Hence, we obtained another structure that is similar to the original configuration space, but physically distinct. Since there are infinitely many choices for $\wh{S}$, there are infinitely many physically distinct ways to choose the configuration space of positions, and therefore the $3$D-space.

\textbf{II)} \emph{Recovering the $3$D-space from the Hamiltonian and the state vector.}

Requiring $\wh{S}$ to not only commute with $\wh{H}$, but also to preserve the state vector $\ket{\psi}$, reduces the possible choices for $\wh{S}$, but there are still infinitely many such choices.
Since $\hilbert\cong L^2(\R^{3\n},\C)$, the space of square integrable complex functions on $\R^{3\n}$, the eigenspaces of the Hamiltonian have infinite dimension. Then, on each such eigenspace, there are infinitely many unitary transformations that commute with $\wh{H}$ and preserve the projection of $\ket{\psi}$ on that eigenspace. Since a unitary transformation $\wh{S}$ that commutes with $\wh{H}$ consists of unitary transformations of the eigenspaces of $\wh{H}$, it follows that there are infinitely many different choices for $\wh{S}$, even if it is required to preserve $\ket{\psi}$.

Therefore, if there is a set of position-like operators $(\wh{\wt{\xThree}}_j)_j$, there are infinitely many such sets, defining distinct $3$D-space structures.

Now suppose that, to achieve uniqueness, we declare all of them physically equivalent, so that the $3$D-space structure is identified with their equivalence class. Then, in particular, the equivalence class of $3$D-spaces will be independent of time, and this would mean that it would be impossible to observe changes of the wavefunction with respect to the $3$D-space. So we cannot take the equivalence class as representing the $3$D-space structure, and therefore the $3$D-space structure cannot be unique.
\end{proof}

\begin{remark}
\label{rem:nrqm-space-isometries}
Choosing $\wh{S}$ as in eq. \eqref{eq:S_time_evolution} is general enough to provide examples of distinct $3$D-space structures that cannot be related by coordinate changes of the $3$D-space, because in general the wavefunction at different times differ more than by a $3$D-space isometry.
\end{remark}

\begin{remark}
\label{rem:nrqm-classical}
In NRQM, there is a unique $3$D-space structure, or, equivalently, a unique set of variables representing classical positions and also momenta (up to coordinate changes of the $3$D-space). But if we eliminate this information and we keep only the {\MQS}, when we try to recover them from the {\MQS} we obtain infinitely many possibilities that are not distinguishable, by the mathematical relations they satisfy, from the ones representing ``the actual'' positions and momenta.

An example when we do not know the position and momentum operators is if we have a different theory, based for example on information or probabilities as in various axiomatic approaches, and we recovered the Hilbert space from the axioms.
Suppose we want to recover NRQM with the Hamiltonian \eqref{eq:schrod_hamiltonian_NRQM}, but we don't know how to identify the $3$D-space in our theory.
We can characterize the $3$D-space structure by operators that have the same properties as the position operators in NRQM, and we require the Hamiltonian to be \eqref{eq:schrod_hamiltonian_NRQM} in terms of the position operators and their conjugates (see Remark \ref{rem:3dspace:nrqm:specs}).
Theorem \ref{thm:NRQM_nogo} shows that there are infinitely many different sets of operators that behave just like the position operators, and so that the Hamiltonian is expressed in terms of them and their conjugates just like in \eqref{eq:schrod_hamiltonian_NRQM}.
The way to prove it is by finding unitary transformations that preserve the Hamiltonian.

Therefore, the whole point of this article is that the {\MQS} has to be supplemented with the physical meaning of the structures, because there is no other way to choose a ``preferred'' one.
\end{remark}

\begin{remark}
\label{rem:nrqm_why_S_H_commute}
In the proof of Theorem \ref{thm:NRQM_nogo}, we assumed that $\wh{S}$ commutes with $\wh{H}$. Without making this assumption, it could be objected that it is not a surprise that the resulting $3$D-space structure is different, since the Hamiltonian $\wh{S}\wh{H}\wh{S}^{-1}$ is different from $\wh{H}$.
But if $[\wh{S},\wh{H}]=0$, $\wh{S}\wh{H}\wh{S}^{-1}=\wh{H}$, and the same objection does not apply to Theorem \ref{thm:NRQM_nogo}.
Moreover, in the cases when the condition $[\wh{S},\wh{H}]=0$ can be dropped, for example when the preferred structures are claimed to depend only on the Hamiltonian's spectrum, as in \cite{CotlerEtAl2019LocalityFromSpectrum,CarrollSingh2019MadDogEverettianism,Carroll2021RealityAsAVectorInHilbertSpace} (discussed in \sref{s:proof-TPS-space}), there will be even more physically distinct structures of the same kind, as we will see in Remark \ref{rem:commute}.
\end{remark}

\section{Non-uniqueness of general preferred structures}
\label{s:proof-TS}

To obtain a general proof that applies to all kinds of structures, including generalized basis, tensor product structure, and emerging $3$D-space structure, we need some preparation.
First, we need to define the most general notion of invariant structure on a Hilbert space.
Then, to be able to verify that a structure is ``essentially unique'' in its kind, we will explain how to characterize the situation when the differences between two structures of the same kind do not matter. This is when they are related by ``non-physical'' transformations, like phase or gauge transformations.
And finally, we need to characterize those invariant structures that are physically relevant.
Then, the main theorem will show that if a structure is essentially unique in its kind, it is physically irrelevant.

These results are fully general, they apply whenever in an approach it is expected, for whatever reasons, that out of many structures of a particular kind, an essentially unique one emerges as preferred, so that it corresponds to a unique structure that we observe in reality. Since the proof covers all imaginable situations, we will not be concerned with particular examples in this Section, but several such examples will be discussed in Sec. \sref{s:proof-TS-applications}.

\subsection{General kinds of structures}
\label{s:proof-TS:kinds}

First, we need to characterize a structure. We will do this by specifying its ``kind'', which will be explained. 
The symmetry of the {\MQS} requires us to define such structures as tensor objects over the Hilbert space, and the kinds of the structures as the types of these tensor objects plus unitary invariant conditions that they are required to satisfy \cite{Weyl1946ClassicalGroupsInvariantsAndRepresentations,Rota2001WhatIsInvariantTheoryReally}. The conditions are needed to express what it means to be a ``preferred structure'' of a particular kind.

We denote the space of \emph{tensors} of type $(r,s)$ over $\hilbert$, \ie the space of $\C$-\emph{multilinear} functions from $\botimes^r\hilbert^\ast\otimes\botimes^s\hilbert$ to $\C$, where $\hilbert^\ast$ is the dual of $\hilbert$, by
\begin{equation}
\label{eq:tensors}
\mc{T}^r_s(\hilbert):=\botimes^r\hilbert\otimes\botimes^s\hilbert^\ast.
\end{equation}

The tensor algebra of the Hilbert space $\hilbert$ is
\begin{equation}
\label{eq:tensor_algebra}
\mc{T}(\hilbert):=\boplus_{r=0}^{\infty}\boplus_{s=0}^{\infty}\mc{T}^r_s(\hilbert).
\end{equation}

If $\dwh{A}\in\mc{T}^r_s(\hilbert)$ is a tensor, and $\wh{S}$ is a unitary transformation of $\hilbert$, we denote by $\wh{S}[\dwh{A}]$ the tensor obtained by unitary transformation from the tensor $\dwh{A}$. In particular, scalars $c\in\mc{T}^0_0(\hilbert)\cong\C$ are invariant constants $\wh{S}[c]=c$. For $\ket{\psi}\in\mc{T}^1_0(\hilbert)=\hilbert$, $\wh{S}[\ket{\psi}]=\wh{S}\ket{\psi}$, $\wh{S}[\bra{\psi}]=\bra{\psi}\wh{S}^\dagger$, and for $\wh{A}\in\mc{T}^1_1(\hilbert)$, $\wh{S}[\wh{A}]=\wh{S}\wh{A}\wh{S}^\dagger$.
For general tensors we extend the transformation by linearity and by assuming that it commutes with the tensor product, \ie we express the tensor as a linear combination of tensor products of bra and ket vectors, and we transform each of them accordingly, as it is usually done.

We denote by $\Herm\(\hilbert\)\subset\mc{T}^1_1(\hilbert)$ the space of Hermitian operators on $\hilbert$.


\begin{definition}[Tensor structures]
\label{def:tensor}
Let $\mc{A}$ be a set, let $\theta:\mc{A}\to\N^2$ be a function 
\begin{equation}
\label{eq:tensor_type}
\theta(\alpha)=({r_\alpha},{s_\alpha}),
\end{equation}
and let
\begin{equation}
\label{eq:tensor_multitype}
\mc{T}^{\theta(\alpha)}(\hilbert):=\mc{T}^{r_\alpha}_{s_\alpha}(\hilbert)
\end{equation}
for all $\alpha\in\mc{A}$.
Denote by
\begin{equation}
\label{eq:tensor}
\begin{aligned}
\mc{T}^\theta(\hilbert)
:=&\prod_{\alpha\in\mc{A}}\mc{T}^{\theta(\alpha)}=\prod_{\alpha\in\mc{A}}\mc{T}^{r_\alpha}_{s_\alpha}\\
=&\ \{(\dwh{A}_\alpha)_{\alpha\in\mc{A}}\ |\ (\forall\alpha\in\mc{A})\ \dwh{A}_{\alpha}\in\mc{T}^{r_\alpha}_{s_\alpha}(\hilbert)\}\\
\end{aligned}
\end{equation}
the set of all structures consisting of tensors
\begin{equation}
\label{eq:tensor_A}
(\dwh{A}_\alpha)_{\alpha\in\mc{A}},
\end{equation}
where $\dwh{A}_{\alpha}\in\mc{T}^{\theta(\alpha)}(\hilbert)$ for all $\alpha\in\mc{A}$, and $\prod$ stands for the Cartesian product.
We call the elements of $\mc{T}^\theta(\hilbert)$ \emph{tensor structures} of type $\theta$.
\end{definition}

\begin{definition}[Invariant tensor functions]
\label{def:invariant_tensor_function}
Let $\mc{A}$ be a set and let $\theta:\mc{A}\to\N^2$ a tensor structure type.
An \emph{invariant tensor function} is a function
\begin{equation}
\label{eq:invariant_tensor_function_def}
\mc{F}:\mc{T}^\theta\(\hilbert\)\to\mc{T}(\hilbert)
\end{equation}
which is invariant under unitary symmetries, \ie for any unitary operator $\wh{S}$ on $\hilbert$ and any $(\dwh{A}_\alpha)_{\alpha\in\mc{A}}\in\mc{T}^\theta\(\hilbert\)$,
\begin{equation}
\label{eq:invariant_tensor_function_invariance}
\mc{F}\(\(\wh{S}[\dwh{A}_\alpha]\)_{\alpha}\)=\wh{S}\[\mc{F}\(\(\dwh{A}_\alpha\)_{\alpha}\)\].
\end{equation}
\end{definition}

If $\mc{F}$ is valued in $\R$ or $\C$, it is a \emph{scalar} function, and eq. \eqref{eq:invariant_tensor_function_invariance} becomes
\begin{equation}
\label{eq:invariant_tensor_function_invariance_scalar}
\mc{F}\(\(\wh{S}[\dwh{A}_\alpha]\)_{\alpha}\)=\mc{F}\(\(\dwh{A}_\alpha\)_{\alpha}\).
\end{equation}

In addition to specifying the type of a tensor structure, we also need to specify what conditions it has to satisfy.

Let $\mc{C}$ be an invariant tensor function of tensor structures of type $\theta$
\begin{equation}
\label{eq:constraint}
\mc{C}:\mc{T}^\theta\(\hilbert\)\to\mc{T}(\hilbert).
\end{equation}

The function $\mc{C}$ can be used to define \emph{conditions} to be satisfied by tensor structures of type $\theta$, which are equations or inequations,
\begin{equation}
\label{eq:condition}
\mc{C}\(\(\dwh{A}^{\ket{\psi}}_\alpha\)_{\alpha\in\mc{A}}\)=0 \tn{ }(\tn{or }\geq 0 \tn{ or }>0).
\end{equation}

It makes sense to have a condition as an inequation only if it makes sense to speak about invariant order relations like $\geq 0$ or $>$. This is the case when the function $\mc{C}$ has values in $\R$, but also when its values are self-adjoint operators on $\hilbert$, since then we can talk about positive operators, see for example equation \eqref{eq:POVM_condition_positive_definite}.

To define the kinds of structures that depend only on the \MQS, we need conditions based on invariant functions that depend on the tensor structure, on the Hamiltonian, and possibly on the state vector.

\begin{definition}[Kind]
\label{def:kind}
A \emph{kind} $\kind=\{\mc{C}_{\beta}\}_{\beta\in\mc{B}}$ is a set of conditions defined by invariant tensor functions
\begin{equation}
\label{eq:constraints}
\mc{C}_{\beta}:\mc{T}^\theta\(\hilbert\)\times\Herm(\hilbert)\times\hilbert\to\mc{T}(\hilbert),
\end{equation}
where $\mc{B}$ is a set assumed to be the union of up to three sets, $\mc{B}=\mc{B}_{=}\cup\mc{B}_{\geq}\cup\mc{B}_{>}$.
A tensor structure $\(\dwh{A}_\alpha\)_{\alpha\in\mc{A}}$ is of the kind $\kind=\{\mc{C}_{\beta}\}_{\beta\in\mc{B}}$ if for any $\beta\in\mc{B}$ and $\ket{\psi}\in\hilbert$,
\begin{equation}
\label{eq:defining_conditions}
\mc{C}_{\beta}\(\(\dwh{A}^{\ket{\psi}}_\alpha\)_{\alpha\in\mc{A}},\wh{H},\ket{\psi}\)
\begin{cases}
=0 & \tn{if }\beta\in\mc{B}_{=},\\
\geq 0 & \tn{if }\beta\in\mc{B}_{\geq},\\
> 0 & \tn{if }\beta\in\mc{B}_{>}.\\
\end{cases}
\end{equation}

The factor $\Herm(\hilbert)$ in \eqref{eq:constraints} is needed to allow the functions $\mc{C}_{\beta}$ to depend on the Hamiltonian, and the last factor $\hilbert$ allows them to depend on the state vector.
\end{definition}

The kinds are required to be invariant in order to exclude a symmetry breaking of the {\MQS}, since this would require structures additional to $\wh{H}$ and $\ket{\psi}$.

\begin{definition}[\emph{$\kind$-structure}]
\label{def:TS}
A \emph{structure of \emph{kind}} $\kind=\{\mc{C}_{\beta}\}_{\beta\in\mc{B}}$ (or a \emph{$\kind$-structure}) for the Hamiltonian $\wh{H}$ is defined as a function
\begin{equation}
\label{eq:TS}
\begin{aligned}
&\struct_{\wh{H}}:\hilbert\to\mc{T}^\theta\(\hilbert\),\\
&\struct_{\wh{H}}^{\ket{\psi}}=\(\dwh{A}^{\ket{\psi}}_\alpha\)_{\alpha\in\mc{A}},\\
\end{aligned}
\end{equation}
subject to the \emph{defining conditions} \eqref{eq:defining_conditions}.

The defining conditions are allowed to be functionally independent on some of the tensors $\dwh{A}^{\ket{\psi}}_\alpha$, $\wh{H}$, or $\ket{\psi}$.
\end{definition}

\begin{remark}
\label{rem:hermitian_K_structure}
Definitions \ref{def:kind} and \ref{def:TS} may seem too abstract. Often all of the tensors $\dwh{A}^{\ket{\psi}}_\alpha$ will be Hermitian operators $\wh{A}^{\ket{\psi}}_\alpha$. In this case, we will call the $\kind$-structure \emph{Hermitian $\kind$-structure}. Hermitian $\kind$-structures will turn out to be sufficient for most of the cases discussed in the article. A possible reason why Hermitian operators are sufficient for the relevant cases is that they correspond to observables.
\end{remark}

Let us give simple examples, so that the reader has something concrete in mind when following the proofs.

\begin{example}[Single Hermitian operator]
\label{ex:hermitian_K_structure:single}
A very simple example of $\kind$-structure consists of a single Hermitian operator $\wh{A}$. Is its only defining condition is $\wh{A}^\dagger-\wh{A}=\wh{0}$.

Particular examples include the Hamiltonian operator $\wh{H}$, and projectors, which satisfy the additional condition $\wh{P}^2-\wh{P}=\wh{0}$.
\end{example}

\begin{example}[Basis]
\label{ex:hermitian_K_structure:basis}
A \emph{basis} $\(\ket{\alpha}\)_{\alpha\in\mc{A}}$ of $\hilbert$ can be defined by a Hermitian $\kind$-structure
\begin{equation}
\label{eq:basis_TS}
\struct_{\wh{H}}^{\ket{\psi}}=\(\wh{A}_{\alpha}:=\ket{\alpha}\bra{\alpha}\)_{\alpha\in\mc{A}},
\end{equation}
where the kind $\kind$ is given by the defining conditions
\begin{equation}
\label{eq:basis_TS_conditions}
\begin{cases}
\wh{A}_{\alpha}\wh{A}_{\alpha'}-\wh{A}_{\alpha}\delta_{\alpha\alpha'}=\wh{0},\\
\wh{I}_{\hilbert}-\sum_{\alpha\in\mc{A}}\wh{A}_{\alpha}=\wh{0},\\
\tr\wh{A}_{\alpha}-1=0\\
\end{cases}
\end{equation}
for all $\alpha,\alpha'\in\mc{A}$.
The first condition states that $\wh{A}_{\alpha}$ are projectors on mutually orthogonal subspaces of $\hilbert$,
the second one that they form a complete system,
and the third one that these subspaces are one-dimensional.
\end{example}

In Example \ref{ex:hermitian_K_structure:basis} one does not usually expect the operators $\wh{A}_{\alpha}$ to depend on $\ket{\psi}$, but we may want to consider cases when additional conditions make them dependent. In this case, we will write $\wh{A}^{\ket{\psi}}_{\alpha}$ instead of $\wh{A}_{\alpha}$.
We will see that, even so, there are many physically distinct bases with the same defining conditions.

In Sec. \sref{s:proof-TS-applications} we will see that Definition \ref{def:TS} covers as particular cases
tensor product structures,
more general notions of emergent $3$D-space or spacetime,
and general notions of generalized bases.

\subsection{Conditions of uniqueness and physical relevance}
\label{s:proof-TS:cond}

Let us state the two main conditions that we expect to be satisfied by a preferred $\kind$-structure, and explain their necessity.

\subsubsection{Essential uniqueness}
\label{s:proof-TS:cond-uniqueness}

The first condition we will impose on a $\kind$-structure is to be \emph{unique} or \emph{essentially unique}.
We will not require the structure to be necessarily unique, but we require at least that whenever such a structure is not unique, there are no physical differences between its instances.

\begin{definition}[Physical equivalence]
\label{def:equiv}
Let $\sim$ denote the equivalence relation on $\hilbert$ defined by $\ket{\psi}\sim\ket{\psi'}$ if and only if $\ket{\psi}$ and $\ket{\psi'}$ can be related by phase transformations, gauge symmetries, space isometries \textit{etc}. In other words, $\ket{\psi}$ and $\ket{\psi'}$ represent the same physical state.
In general, physical equivalence is ensured by the action of a group $G_P$ on $\hilbert$, which provides a projective representation of $G_P$. Any element $g$ of $G_P$ is represented by a unitary  or anti-unitary operator $\wh{g}$ on $\hilbert$ that commutes with $\wh{H}$ \cite{Wigner1959GroupTheoryAndItsApplicationToTheQuantumMechanicsOfAtomicSpectra,Bargmann1954OnUnitaryRayRepresentationsOfContinuousGroups}.
Then, $\ket{\psi}\sim\ket{\psi'}$ if and only if there is a group element $g\in G_P$ so that $\wh{g}\ket{\psi}=\ket{\psi'}$.
Since the representation of $G_P$ is linear, it extends uniquely to vectors from $\hilbert^\ast$, $\bra{\psi}\sim\bra{\psi'}$ if and only if there is a group element $g\in G_P$ so that $\bra{\psi}\wh{g}^\dagger=\bra{\psi'}$, and to general tensors on $\hilbert$.
For an element $g\in G_P$ we denote the resulting transformation of a tensor $\dwh{A}$ by $\wh{g}\[\dwh{A}\]$.

Two $\kind$-structures 
$\struct_{\wh{H}}^{\ket{\psi}}=\(\dwh{A}^{\ket{\psi}}_\alpha\)_{\alpha\in\mc{A}}$ 
and 
$\struct'_{\wh{H}}{}^{\ket{\psi}}=\(\dwh{A'}^{\ket{\psi}}_\alpha\)_{\alpha\in\mc{A}}$ are said to be \emph{physically equivalent}, $\struct_{\wh{H}}{}^{\ket{\psi}}\sim\struct'_{\wh{H}}{}^{\ket{\psi}}$, if
there is an element $g\in G_P$ so that
\begin{equation}
\label{eq:uniqueness}
\(\dwh{A'}^{\ket{\psi}}_\alpha\)_{\alpha\in\mc{A}}=\(\wh{g}\[\dwh{A}^{\ket{\psi}}_\alpha\]\)_{\alpha\in\mc{A}}.
\end{equation}

Two state vectors, operators, in general tensors of the same types, or structures of the same kind, are said to be \emph{physically distinct} if they are not physically equivalent.

\end{definition}

\begin{remark}
\label{rem:uniqueness_permutation}
Sometimes, the elements of the physical symmetry group $G_P$ simply permute the tensor structures constituting the $\kind$-structure $\struct_{\wh{H}}^{\ket{\psi}}=\(\dwh{A}^{\ket{\psi}}_\alpha\)_{\alpha\in\mc{A}}$.
That is, the effect of acting with any element $g\in G_P$ as in eq. \eqref{eq:uniqueness} is to permute the set of indices $\mc{A}$.
This is true for most of the situations we will study in this article.
For example, the symmetries of the $3$D-space in NRQM (Sec. \sref{s:examples-NRQM-space}) permute the position operators.
\end{remark}

\begin{condition}[``Essentially unique''-ness]
\label{cond:uniqueness}
Any two $\kind$-structures of the same kind $\kind$ are physically equivalent.
\end{condition}

\begin{observation}
\label{obs:equivalence}
Condition \ref{cond:uniqueness} requires that the symmetry group $G(\kind)$ of the kind $\kind$ is a subgroup of the symmetry group $G_P$ corresponding to the physical equivalence, whatever that group is.
On the other hand, if $G(\kind)\lneq G_P$, any element $g\in G_P\setminus G(\kind)$ would transform a $\kind$-structure into a physically distinct one. Therefore, Condition \ref{cond:uniqueness} requires the two symmetry groups to coincide,
\begin{equation}
\label{eq:kind_symmetries_are_nonphysical}
G(\kind)= G_P.
\end{equation}

It may seem circular to refer to a notion of ``essential uniqueness'' without specifying the equivalence relation $\sim$ or the symmetry group $G_P$.
But without this they need to emerge uniquely from the {\MQS}, or otherwise to be added to $(\hilbert,\wh{H},\ket{\psi})$ as equally fundamental structures.
But in order to avoid counterarguments of the form ``maybe this preferred structure is not unique, but it is essentially unique'', we have to allow this freedom even if we do not define it, for example assuming that it emerges as well. 
We will see that our proof applies regardless of whether $G_P$ is assumed as fundamental, or as emerging from the {\MQS}.
To cover all cases, we simply assume the existence of a group $G_P$ without specifying it, even if it consists of only phase transformations of $\hilbert$, and impose Condition \ref{cond:uniqueness}.
\end{observation}

\begin{remark}
\label{rem:uniqueness}
Consider a $\kind$-structure
\begin{equation}
\label{eq:TSt}
\struct_{\wh{H}}^{\ket{\psi}}=\(\dwh{A}^{\ket{\psi}}_\alpha\)_{\alpha\in\mc{A}}.
\end{equation}
Let $\wh{S}$ be a unitary transformation of $\hilbert$ which commutes with $\wh{H}$.
Then, $\(\dwh{A'}{}^{\wh{S}\ket{\psi}}_\alpha\)_{\alpha\in\mc{A}}$ is a $\kind$-structure for $(\hilbert,\wh{H},\wh{S}\ket{\psi})$, where for each $\alpha\in\mc{A}$
\begin{equation}
\label{eq:TSt_transformed}
\dwh{A'}{}^{\wh{S}\ket{\psi}}_\alpha := \wh{S}[\dwh{A}^{\ket{\psi}}_\alpha].
\end{equation}

The Essential Uniqueness Condition \ref{cond:uniqueness} implies that
\begin{equation}
\label{eq:uniqueness_P}
\struct_{\wh{H}}^{\wh{S}\ket{\psi}} \sim \(\wh{S}[\dwh{A}^{\ket{\psi}}_\alpha]\)_{\alpha\in\mc{A}},
\end{equation}
or, equivalently, the equivalence
\begin{equation}
\label{eq:uniqueness_b}
\(\dwh{A}^{\wh{S}\ket{\psi}}_\alpha\)_{\alpha\in\mc{A}} \sim \(\wh{S}[\dwh{A}^{\ket{\psi}}_\alpha]\)_{\alpha\in\mc{A}},
\end{equation}
whether or not $\wh{S}$ represents an element of $G_P$.
\end{remark}

\begin{proposition}
\label{thm:equiv_A_to_psi}
$\(\dwh{A}^{\ket{\psi}}_\alpha\)_{\alpha\in\mc{A}}$ is a $\kind$-structure for $\ket{\psi}$, if and only if it is a $\kind$-structure for $\wh{g}\ket{\psi}$, for any $g\in G_P$.
\end{proposition}
\begin{proof}
From eq. \eqref{eq:kind_symmetries_are_nonphysical}, $\(\dwh{A}^{\ket{\psi}}_\alpha\)_{\alpha\in\mc{A}}$ is a $\kind$-structure for $\ket{\psi}$ if and only if $\(\wh{g}^{-1}\[\dwh{A}^{\ket{\psi}}_\alpha\]\)_\alpha$ is a $\kind$-structure for $\ket{\psi}$.
By making $\wh{S}=\wh{g}$ in eq. \eqref{eq:invariant_tensor_function_invariance} and applying it to the defining conditions \eqref{eq:defining_conditions} for $\(\wh{g}^{-1}\[\dwh{A}^{\ket{\psi}}_\alpha\]\)_\alpha$ as a $\kind$-structure for $\ket{\psi}$ we obtain that
this is equivalent to $\(\dwh{A}^{\ket{\psi}}_\alpha\)_\alpha$ satisfying the defining conditions for $\wh{g}\ket{\psi}$.
\end{proof}

\subsubsection{Physical relevance}
\label{s:proof-TS:cond-relevance}

The second condition we will impose on a $\kind$-structure is to be \emph{physically relevant}, \ie to be able to distinguish among physically distinct states the system can have, even if the Hamiltonian itself cannot distinguish them.
Without this condition, the $\kind$-structure has no physically observable effect.

\begin{definition}
\label{def:physical_relevance}
A tensor structure $\(\dwh{A}_\alpha\)_{\alpha\in\mc{A}}$ is said to \emph{distinguish} the states $\ket{\psi}$ and $\ket{\psi'}$ if there is an invariant scalar function
\begin{equation}
\label{eq:physical_relevance_invar_scalar}
\mc{I}:\mc{T}^\theta\(\hilbert\)\times\hilbert\to\R,
\end{equation}
\begin{equation}
\label{eq:physical_relevance_invar}
\mc{I}\(\(\dwh{A}_\alpha\)_{\alpha\in\mc{A}},\ket{\psi}\)
\neq\mc{I}\(\(\dwh{A}_\alpha\)_{\alpha\in\mc{A}},\ket{\psi'}\).
\end{equation}
\end{definition}

\begin{example}
\label{ex:physical_relevance_operators}
A Hermitian operator $\wh{A}$ (\cf Example \ref{ex:hermitian_K_structure:single}) can distinguish two unit eigenvectors $\ket{\psi}$ and $\ket{\psi'}$ corresponding to distinct eigenvalues $\lambda\neq\lambda'$, because $\bra{\psi}\wh{A}\ket{\psi}=\lambda\neq\lambda'=\bra{\psi'}\wh{A}\ket{\psi'}$.
In general, $\wh{A}$ distinguishes two unit vectors $\ket{\psi}$ and $\ket{\psi'}$ when the projections of $\ket{\psi}$ and $\ket{\psi'}$ on at least an eigenspace of $\wh{A}$ have different norms.
In particular, if the mean value differs for the two states, $\meanvalue{\wh{A}}{\ket{\psi}}\neq\meanvalue{\wh{A}}{\ket{\psi'}}$, $\wh{A}$ distinguishes $\ket{\psi}$ and $\ket{\psi'}$.
From the point of view of quantum measurements, distinguishability is perfect only when $\braket{\psi}{\psi'}=0$.
These observations apply in particular to the Hamiltonian operator.
\end{example}

\begin{proposition}
\label{thm:orbit}
A Hermitian operator $\wh{A}$ distinguishes two unit vectors $\ket{\psi}$ and $\ket{\psi'}$ if and only if there is no unitary $\wh{S}$ commuting with $\wh{A}$ so that $\ket{\psi'}=\wh{S}\ket{\psi}$. 
\end{proposition}
\begin{proof}
Let $\mc{I}:\Herm\otimes\hilbert\to\R$ be an invariant scalar function.
Suppose that $\ket{\psi'}=\wh{S}\ket{\psi}$.
Then,
\begin{equation}
\label{eq:orbit_a}
\begin{aligned}
\mc{I}\(\wh{A},\ket{\psi'}\)
&\ =\mc{I}\(\wh{A},\wh{S}\ket{\psi}\) \\
&\stackrel{\eqref{eq:invariant_tensor_function_invariance_scalar}}{=}\mc{I}\(\wh{S}^{-1}\wh{A}\wh{S},\ket{\psi}\). \\
\end{aligned}
\end{equation}

If $[\wh{S},\wh{A}]=0$, then $\wh{S}^{-1}\wh{A}\wh{S}=\wh{A}$ and eq. \eqref{eq:orbit_a} becomes
\begin{equation}
\label{eq:orbit_a_commute}
\mc{I}\(\wh{A},\ket{\psi'}\)=\mc{I}\(\wh{A},\ket{\psi}\).
\end{equation}

Therefore, if $[\wh{S},\wh{A}]=0$, there is no invariant scalar function $\mc{I}$ so that eq. \eqref{eq:physical_relevance_invar} is true.
Conversely, it follows that if there is an invariant scalar function so that eq. \eqref{eq:physical_relevance_invar} is true, then there is no unitary operator $\wh{S}$ commuting with $\wh{A}$ so that $\ket{\psi'}=\wh{S}\ket{\psi}$.
\end{proof}

The $\kind$-structure of interest has to be physically relevant. This means that it has to be able to distinguish at least some pairs of physically distinct states. Since we want to allow some freedom for our structures and require uniqueness to be only up to non-physical degrees of freedom as in Condition \ref{cond:uniqueness}, we have to take this freedom into account in the condition of physical relevance.

\begin{condition}[Physical relevance]
\label{cond:physical_relevance}
There exist at least two unit vectors $\ket{\psi}\nsim\ket{\psi'}\in\hilbert$ representing distinct physical states that are not distinguished by the Hamiltonian,
and an invariant scalar function $\mc{I}$ of $\struct_{\wh{H}}^{\ket{\psi}}$ and $\ket{\psi}$ so that, for any elements $g,g'\in G_P$,
\begin{equation}
\label{eq:physical_relevance_invar_G}
\mc{I}\(\wh{g}\[\struct_{\wh{H}}^{\ket{\psi}}\],\ket{\psi}\)
\neq\mc{I}\(\wh{g'}\[\struct{}_{\wh{H}}^{\ket{\psi'}}\],\ket{\psi'}\).
\end{equation}
\end{condition}

\begin{remark}
\label{rem:cond:physical_relevance}
Condition \ref{cond:physical_relevance} may seem a bit too complicated, so let us explain it. If every physical state would be represented by a unique ray in the Hilbert space, eq. \eqref{eq:physical_relevance_invar_G} would just require that the $\kind$-structure $\struct_{\wh{H}}$ distinguishes any two linearly independent state vectors $\ket{\psi}$ and $\ket{\psi'}$. But if a physical state can be represented by more linearly independent state vectors, we need to impose the condition of physical relevance for physically distinct states, rather than just for distinct rays in the Hilbert space.
This requires the condition to hold for equivalence classes of $\kind$-structures.

Condition \ref{cond:physical_relevance} also requires that $\ket{\psi}$ and $\ket{\psi'}\in\hilbert$ are not distinguished by the Hamiltonian. The reason is that, if the Hamiltonian is able to distinguish them, then it may be unclear whether the structure $\struct_{\wh{H}}^{\ket{\psi}}$ distinguishes them by itself, or because of its dependence of the Hamiltonian.

Let us take as an example $\wh{S}=\wh{U}_{t,t_0}$ as in eq. \eqref{eq:S_time_evolution}. We expect $\ket{\psi(t_0)}$ and $\ket{\psi(t)}=\wh{U}_{t,t_0}\ket{\psi(t_0)}$ to be physically distinct, since in general the system changes in time.
The Hamiltonian does not distinguish them, because the projections of $\ket{\psi(t)}$ on each eigenspace of $\wh{H}$ does not change in time. 
But interesting structures like the $3$D-space as in Theorem \ref{thm:NRQM_nogo} should be able to distinguish them.
Moreover, they should be able to distinguish $\wh{g}\ket{\psi(t_0)}$ and $\wh{g}'\ket{\psi(t)}$ even when isometries of space or gauge symmetries, represented by $\wh{g}$ and $\wh{g}'$, are applied.
Therefore, there are infinitely many physically distinct states obtained by unitary evolution from the same initial state.

Not only unitary operators equal to time evolution operators relate distinguishable states. There are many other families of such unitary operators, in fact infinitely many independent such families, if the dimension of the Hilbert space is infinite (see Remark \ref{rem:distinguishing_special}).
Condition \ref{cond:physical_relevance} captures the idea, used in the proof of Theorem \ref{thm:NRQM_nogo}, that different physical states ``look different'' with respect to a candidate preferred structure.

The differences have to be objective, in the sense that they are not due to a particular choice of the basis, because they have to be physically observable.

This should be true for a candidate preferred space, basis, and tensor product structure. Without this condition, the structure would be physically irrelevant.
\end{remark}

\subsection{Uniqueness contradicts physical relevance}
\label{s:proof-TS:thm}

We will now prove that there is a contradiction between the Conditions \ref{cond:uniqueness} and \ref{cond:physical_relevance}.

\begin{theorem}
\label{thm:nogo}
If a $\kind$-structure is physically relevant, then it is not (essentially) unique.
\end{theorem}
\begin{proof}
Let $\wh{S}$ be a unitary operator so that $\ket{\psi'}=\wh{S}\ket{\psi}$.
From Proposition \ref{thm:orbit}, we can choose $\wh{S}$ to commute with $\wh{H}$.
The $\kind$-structure $\struct_{\wh{H}}^{\ket{\psi}}$ being a tensor object, its defining conditions are invariant to any unitary transformation $\wh{S}$ that commutes with $\wh{H}$. 
Then $\wh{S}^{-1}$ \emph{actively} transforms $\struct_{\wh{H}}^{\wh{S}\ket{\psi}}$ into a $\kind$-structure $\struct_{\wh{H}}^{'\ket{\psi}}$ for $\ket{\psi}$,
\begin{equation}
\label{eq:TSt_transformed_inv}
\struct_{\wh{H}}^{'\ket{\psi}}:=\wh{S}^{-1}\[\struct_{\wh{H}}^{\wh{S}\ket{\psi}}\].
\end{equation}

The transformation $\wh{S}^{-1}$ has to be active, because we want to find other possible structures of the same kind, and not just to express the same structure in a different basis (see Remark \ref{rem:active-not-passive} and {\lessonName} \ref{lesson:invariance-fix}).

From the uniqueness condition \eqref{eq:uniqueness}, \eqref{eq:TSt_transformed_inv} implies that
\begin{equation}
	\label{eq:uniqueness_unitary_inv}
	\wh{g}^{-1}\[\struct_{\wh{H}}^{\ket{\psi}}\] \stackrel{\eqref{eq:uniqueness}}{=} \wh{S}^{-1}\[\struct_{\wh{H}}^{\wh{S}\ket{\psi}}\]
\end{equation}
for some transformation $\wh{g}\in G_P$.
From equations \eqref{eq:uniqueness_unitary_inv} and \eqref{eq:invariant_tensor_function_invariance_scalar}, for any invariant function $\mc{I}$,
\begin{equation}
	\label{eq:uniqueness_unitary_inv_psi}
\begin{aligned}
\mc{I}\(\wh{g}^{-1}\[\struct_{\wh{H}}^{\ket{\psi}}\],\ket{\psi}\)
&\stackrel{\eqref{eq:uniqueness_unitary_inv}}{=}
\mc{I}\(\wh{S}^{-1}\[\struct_{\wh{H}}^{\wh{S}\ket{\psi}}\],\ket{\psi}\) \\
&\stackrel{\eqref{eq:invariant_tensor_function_invariance_scalar}}{=}
\mc{I}\(\struct_{\wh{H}}^{\wh{S}\ket{\psi}},\wh{S}\ket{\psi}\) \\
&\ =\ 
\mc{I}\(\struct_{\wh{H}}^{\wh{S}\ket{\psi}},\ket{\psi'}\).\\
\end{aligned}
\end{equation}

This contradicts Condition \ref{cond:physical_relevance}.
Therefore, Conditions \ref{cond:uniqueness} and \ref{cond:physical_relevance} cannot both be true.
\end{proof}

\begin{remark}
\label{rem:commute}
In the proof of Theorem \ref{thm:nogo}, the condition that $\wh{S}$ commutes with $\wh{H}$ allowed $\wh{S}^{-1}$ to transform $\struct_{\wh{H}}^{\wh{S}\ket{\psi}}$ into a $\kind$-structure $\struct_{\wh{H}}^{'\ket{\psi}}$ for the {\MQS} $(\hilbert,\wh{H},\ket{\psi})$ as in eq. \eqref{eq:TSt_transformed_inv}.
Without $[\wh{S},\wh{H}]=0$, $\wh{S}^{-1}$ would transform $\struct_{\wh{H}}^{\wh{S}\ket{\psi}}$ into a $\kind$-structure for the {\MQS} $(\hilbert,\wh{S}^{-1}\wh{H}\wh{S},\ket{\psi})$ with $\wh{S}^{-1}\wh{H}\wh{S}\neq\wh{H}$.
Then, it could be objected that it is expected to obtain a different $\kind$-structure, since the Hamiltonian $\wh{S}^{-1}\wh{H}\wh{S}$ is different from $\wh{H}$ (also see Remark \ref{rem:nrqm_why_S_H_commute}).

But for structures supposed to be determined by the Hamiltonian's spectrum, as in \cite{CotlerEtAl2019LocalityFromSpectrum,CarrollSingh2019MadDogEverettianism,Carroll2021RealityAsAVectorInHilbertSpace}, we can also use unitary transformations $\wh{S}$ that do not commute with $\wh{H}$, because they preserve the spectrum of $\wh{H}$.
This will give even more physically distinct structures than Theorem \ref{thm:nogo} already does.
In general, the less our structure is constrained by the Hamiltonian, the more physically distinct structures of the same kind exist.
\end{remark}

\begin{remark}
\label{rem:unique}
There are ways to construct structures that depend on the {\MQS} alone and are unique, but they all violate the Physical Relevance Condition \ref{cond:physical_relevance}.
Such examples include the trivial ones $\ket{\psi}$, $\bra{\psi}$, $\ket{\psi}\bra{\psi}$, $\wh{H}$, $\wh{H}\ket{\psi}$, 
more general operators like $f(\wh{H})$, where $f(x)$ is a formal polynomial or formal power series,
but also the direct sum decomposition of the Hilbert space into eigenspaces of $\wh{H}$, projections of $\ket{\psi}$ on these eigenspaces \textit{etc}. In general, any invariant tensor function (\cf Definition \ref{def:invariant_tensor_function}) $\mc{F}(\wh{H},\ket{\psi})$, where $\mc{F}:\Herm(\hilbert)\times\hilbert\to\mc{T}(\hilbert)$, leads to a unique structure. Theorem \ref{thm:nogo} does not apply to such structures, because they violate Condition \ref{cond:physical_relevance}.
The applications of Theorem \ref{thm:nogo} are limited to those structures that satisfy Condition \ref{cond:physical_relevance}.
\end{remark}

\begin{remark}
\label{rem:others}
There are known examples of simple Hamiltonians for which the position and momentum operators are not uniquely obtainable from the Hamiltonian \cite{Schmelzer2009WhyTheHamiltonOperatorAloneIsNotEnough,Schmelzer2011PureQuantumInterpretationsAreNotViable}. But these examples are very particular. By contrast, Theorem \ref{thm:nogo} is fully general, and applies to any possible Hamiltonian. 
\end{remark}

\begin{remark}
\label{rem:versatility}
In another article \cite{Stoica2022VersatilityOfTranslations} we provide an example of {\MQS} which admits infinitely many physically distinct space structures and factorizations. In this {\MQS}, the Hamiltonian is $\wh{H}=-i\hbar\frac{\partial\ }{\partial\tau}$, and for any $t$, $\ket{\psi(t)}$ is an eigenstate of its conjugate operator $\wh{\tau}$. Three infinite lists of examples were found to have this {\MQS}:

1. The \emph{quantum representations} (discovered by Koopman and von Neumann \cite{Koopman1931HamiltonianSystemsAndTransformationInHilbertSpace,vonNeumann1932KoopmanMethod}) of deterministic time-reversible dynamical systems without time loops.
Since there are infinitely many such dynamical systems, with infinitely many choices of space, of the decomposition into subsystems, and of other physical properties, we get an infinite family of physically distinct quantum systems with the same {\MQS}.

2. The measuring device and the observed system in the standard model of ideal measurements.

3. Any closed quantum system which contains an isolated subsystem with the Hamiltonian $\wh{H}=-i\hbar\frac{\partial\ }{\partial\tau}$, for example an ideal isolated clock or a sterile massless fermion in a certain eigenstate.
\end{remark}

These results confirm the conclusion of Theorem \ref{thm:nogo}, that the {\MQS} is not sufficient, and some of the observables have to be assumed from start to represent certain physical properties.

Theorem \ref{thm:nogo} affects all approaches to QM that rely on the {\MQS} alone, including those where state vector reduction takes place.

\subsection{Generalization to C*-algebras}
\label{s:proof-TS:thm-C-star}

Theorem \ref{thm:nogo} already covers all possible attempts to recover a preferred and physically relevant structure from any {\MQS} $\bigl(\hilbert,\wh{H},\ket{\psi}\bigr)$. The proof applies to all Hilbert spaces, being them finite or infinite-dimensional, separable or not.

However, one may imagine that it is possible to circumvent Theorem \ref{thm:nogo} by working in a more general paradigm. Since the most general such formalism is provided by the C*-algebras (which also include the von Neumann algebras), let us see how the notion of {\MQS} and Theorem \ref{thm:nogo} generalize to C*-algebras \cite{Segal1947C-star-algebras-IrreducibleRepresentationsOfOperatorAlgebras,Strocchi2008IntroductionToTheMathematicalStructureOfQuantumMechanics}.

\begin{definition}
\label{def:C-star}
A \emph{Banach algebra} is an associative algebra $\ms{A}$ over $\R$ or $\C$, which is a normed vector space complete in the metric induced by the norm, so that $\norm{\wh{A}\wh{B}}\leq\norm{\wh{A}}\norm{\wh{B}}$ for all $\wh{A},\wh{B}\in\ms{A}$.
A \emph{C*-algebra} is a complex Banach algebra $\ms{A}$ endowed with an anti-linear \emph{involution map} $\ast:\ms{A}\to\ms{A}$, where $\bigl(\wh{A}^\ast\bigr)^\ast=\wh{A}$, $\bigl(\wh{A}+\wh{B}\bigr)^\ast=\wh{A}^\ast+\wh{B}^\ast$, $\bigl(\wh{A}\wh{B}\bigr)^\ast=\wh{B}^\ast \wh{A}^\ast$, $\bigl(c\wh{A}\bigr)^\ast=c^\ast \wh{A}^\ast$, and $\norm{\wh{A}^\ast \wh{A}}=\norm{\wh{A}}\norm{\wh{A}^\ast}$, for all $\wh{A},\wh{B}\in\ms{A}$ and $c\in\C$.

A \emph{*-homomorphism} between two C*-algebras $\ms{A}$ and $\ms{B}$ is a bounded linear map $\alpha:\ms{A}\to\ms{B}$ so that $\alpha\bigl(\wh{A}\wh{B}\bigr)=\alpha\bigl(\wh{A}\bigr)\alpha\bigl(\wh{B}\bigr)$ and $\alpha\bigl(\wh{A}^\ast\bigr)=\alpha\bigl(\wh{A}\bigr)^\ast$ for all $\wh{A},\wh{B}\in\ms{A}$. A bijective *-homomorphism is called \emph{*-isomorphism}. 
A *-isomorphism from $\ms{A}$ to $\ms{A}$ is called \emph{*-automorphism}.
\end{definition}

The bounded linear operators on a Hilbert space $\hilbert$ form a C*-algebra, allowing C*-algebraic formulations of Quantum Theory.
In C*-algebraic formulations it is possible to avoid any reference to the Hilbert space $\hilbert$ and refer only to the operators.
The states correspond to \emph{normalized positive linear functionals} on $\ms{A}$, for example the density operator $\wh{\rho}:=\ket{\psi}\bra{\psi}$ of a state vector $\ket{\psi}$ acts as a normalized positive linear functional by
\begin{equation}
\label{eq:trace}
\wh{A}\mapsto\tr\bigl(\wh{\rho}\wh{A}\bigr)\in\R.
\end{equation}

If the Hamiltonian $\wh{H}$ is unbounded, which is the case for example in NRQM and Quantum Field Theory, it cannot correspond to an element of a C*-algebra.
But it can be represented equivalently as the one-parameter group of unitary evolution operators $\bigl(\wh{U}(t)\bigr)_{t\in\R}$. In the case when $\hilbert$ is finite-dimensional or separable, Stone's theorem \cite{Stone1932OnOneParameterUnitaryGroupsInHilbertSpace} shows that the Hamiltonian is equivalent to the one-parameter group of unitary evolution operators it generates. These operators are bounded and define *-automorphisms.

In C*-algebraic settings the dynamics can be defined in terms of a one-parameter group of *-automorphisms $\bigl(\wh{U}(t)\bigr)_{t\in\R}$ instead of the Hamiltonian \cite{Strocchi2008IntroductionToTheMathematicalStructureOfQuantumMechanics}.
Therefore, the natural C*-algebraic equivalent of the {\MQS} $\bigl(\hilbert,\wh{H},\ket{\psi}\bigr)$ is
\begin{equation}
\label{eq:MQS-C-star}
\bigl(\ms{A},(\wh{U}(t))_{t\in\R},\rho\bigr),
\end{equation}
where $\ms{A}$ is a C*-algebra and $\rho$ is a normalized positive linear functional on $\ms{A}$.
All invariant conditions involving the Hamiltonian $\wh{H}$ can be replaced with invariant conditions involving the one-parameter group of *-automorphisms $\bigl(\wh{U}(t)\bigr)_{t\in\R}$.

The position and momentum operators $\wh{q}$ and $\wh{p}$, satisfying the \emph{canonical commutation relations} (CCR)
\begin{equation}
\label{eq:CCR}
[\wh{q},\wh{p}]=i\hbar\wh{I},
\end{equation}
are unbounded as well.
But they can also be expressed in terms of unbounded operators, the \emph{Weyl operators} \cite{Hall2013QuantumTheoryForMathematicians}
\begin{equation}
\label{eq:Weyl-operators}
\wh{P}(a):=\ee^{ia\wh{p}}, \wh{Q}(b):=\ee^{ib\wh{q}},
\end{equation}
for all $a,b\in\R$.
The CCR \eqref{eq:CCR} becomes, for all $a,b\in\R$, \emph{Weyl's braiding relation}
\begin{equation}
\label{eq:CCR-C-star}
\wh{P}(a)\wh{Q}(b)=\ee^{-iab}\wh{Q}(b)\wh{P}(a).
\end{equation}

The CCR \eqref{eq:CCR} do not admit finite-dimensional representations, although the \emph{clock} and \emph{shift} matrices satisfy approximate CCR relations and converge to \eqref{eq:CCR} in the infinite dimension limit \cite{Weyl1927QuantenmechanikUndGruppentheorie,Santhanam1976QuantumMechanicsInFiniteDimensions,Hall2013QuantumTheoryForMathematicians}. By the Stone-von Neumann theorem \cite{vonNeumann1932UberEinenSatzVonHerrnMHStone}, all irreducible Hilbert space representations of the CCR \eqref{eq:CCR} are unitarily equivalent, and require an infinite-dimensional separable Hilbert space.

The relations \eqref{eq:CCR-C-star} also work in the C*-algebraic formalism, and are more general than the CCR \eqref{eq:CCR}, because they work in the nonregular cases, when the $\wh{q}$ and $\wh{p}$ operators cannot be recovered as self-adjoint operators that satisfy the CCR \eqref{eq:CCR}.

A classification of the representations of the C*-algebraic CCR \eqref{eq:CCR-C-star} up to unitary equivalence, which generalizes the Stone-von Neumann theorem to strongly measurable, but not necessarily strongly continuous representations, is provided in \cite{CavallaroMorchioStrocchi1999GeneralizationStoneVonNeumannTheoremToNonregularRepresentationsCCRAlgebra}.
It includes all the (nonregular) representations considered in physical models.

\begin{remark}
\label{rem:GNS}
The apparent gain in generality obtained by replacing the {\MQS} \eqref{eq:MQS} with \eqref{eq:MQS-C-star} cannot avoid Theorem \ref{thm:nogo}.

It is not difficult to adapt Definition \ref{def:kind}, Conditions \ref{cond:uniqueness} \&  \ref{cond:physical_relevance}, and Theorem \ref{thm:nogo} to quantum theories in the C*-algebraic framework. For example, since all tensor structures  that we encounter in this article are Hermitian $\kind$-structures, all we have to do is to replace the Hermitian operators from the tensor structure with one-parameter groups like we did for the Hamiltonian and the position and momentum operators.
Then, we rewrite the defining conditions for the $\kind$-structures in terms of one-parameter groups consisting of elements of $\ms{A}$. Then, Theorem \ref{thm:nogo} can be proved for C*-algebras exactly how it was proved in this article.

Another way to see this, which has the advantages of being direct and not limited to Hermitian $\kind$-structures, is to recall that the \emph{Gelfand-Naimark-Segal} (GNS) construction recovers the Hilbert space formulation as a canonical representation of the C*-algebra $\ms{A}$ \cite{Segal1947C-star-algebras-IrreducibleRepresentationsOfOperatorAlgebras}. 
Therefore, the GNS construction recovers a {\MQS} as in eq. \eqref{eq:MQS} from the data from eq. \eqref{eq:MQS-C-star}.
Then, Theorem \ref{thm:nogo}, which applies to any Hilbert space, can be applied to the GNS representation as well.

As an example of application, consider the case of recovering the $3$D-space by using position and momentum operators.
In addition to requiring the CCR to be satisfied, we can complete the defining conditions of the $\kind$-structure with conditions that specify the representation of the CCR. These conditions have to be invariant as well to unitary transformations, so Theorem \ref{thm:nogo} applies.

Therefore, the main result of this article extends to C*-algebra formulations as well.
\end{remark}

\section[The abundance of physically relevant structures of\\\mbox{} the same kind]{The abundance of physically relevant structures of the same kind}
\label{s:abundance}

\begin{remark}
\label{rem:distinguishing_special}
Even if in Condition \ref{cond:physical_relevance} we modestly require that the $\kind$-structure $\struct_{\wh{H}}$ distinguishes at least two state vectors, there are in fact many infinite families of distinct structures.
Every self-adjoint operator that commutes with the Hamiltonian $\wh{H}$ generates a one-parameter group of unitary operators that commute with $\wh{H}$. If the dimension of the Hilbert space is infinite, since infinitely many linearly independent self-adjoint operators commute with $\wh{H}$, there are infinitely many independent one-parameter groups of unitary operators that commute with $\wh{H}$, leading to infinitely many families of distinct structures of the same kind.
\end{remark}

\begin{example}
\label{ex:distinguishing_special_time_evolution}
One of these groups of unitary operators that commute with $\wh{H}$ is generated by $\wh{H}$ itself.
Consider the \emph{unitary time evolution operator}, 
\begin{equation}
\label{eq:unitary_evolution_operator}
\wh{U}_{t,t_0}:=\ee^{-\frac{\ii}{\hbar}\wh{H}(t-t_0)},
\end{equation}
which satisfies
\begin{equation}
\label{eq:unitary_evolution}
\ket{\psi(t)}=\wh{U}_{t,t_0}\ket{\psi(t_0)}.
\end{equation}

The operator $\wh{U}_{t,t_0}$ commutes of course with the Hamiltonian.
And unless $\ket{\psi(t_0)}$ is an eigenstate of $\wh{H}$, the physical state of the system changes in time. In general, the states at different times are physically distinct.
\end{example}

\begin{definition}[Time-distinguishing structure]
\label{def:physical_relevance-time}
A tensor structure $\bigl(\dwh{A}^{\ket{\psi}}_\alpha\bigr)_{\alpha\in\mc{A}}$ is said to be \emph{time-distinguishing} for $\ket{\psi_1},\ket{\psi_2}\in\hilbert$ if it distinguishes them and they are connected by unitary evolution, \ie $\ket{\psi_2}=\wh{U}_{t_2,t_1}\ket{\psi_1}$ for some $t_1,t_2\in\R$.
According to Example \ref{ex:distinguishing_special_time_evolution}, one can take $\wh{S}=\wh{U}_{t,t_0}^{-1}$ as a unitary transformation of $\hilbert$ which transforms $\ket{\psi_1}$ in a state vector equal to $\ket{\psi_2}$ \emph{at the same time}, so time-distinguishingness is just a particular form of distinguishingness.
\end{definition}

Time-distinguishingness is always between states that are not distinguished by the Hamiltonian, because $[\wh{U}_{t_2,t_1},\wh{H}]=0$ (see Example \ref{ex:physical_relevance_operators}).

\begin{corollary}
\label{thm:nogo_time}
If a $\kind$-structure is time-distinguishing ({\cf} Definition \ref{def:physical_relevance-time}), then it is not unique.
\end{corollary}
\begin{proof}
It is immediate from Theorem \ref{thm:nogo}, since time-distinguishingness is a form of distinguishingness.
The particular case of time-distinguishingness follows by taking $\wh{S}=\wh{U}_{t_2,t_1}^{-1}$ (see Fig. \ref{sketch-proof.pdf}).
\end{proof}

\begin{remark}
\label{rem:nogo_time}
Corollary \ref{thm:nogo_time} illustrates why Condition \ref{cond:physical_relevance} requires, for physical relevance, that states not distinguished by the Hamiltonian itself can be distinguished. It also  provides easy to see cases of physical relevance, since we expect that our emergent structures distinguish changes of the state of the system in time.
\end{remark}

Time-distinguishingness is useful because it shows the abundance of non-unique structures of the same kind, provided that they distinguish distinct physical states.
But in this article we avoided using it as a criterion of physical relevance, and preferred instead the most general type distinguishingness provided by Remark \ref{rem:distinguishing_special}.

Moreover, in \cite{Stoica2023PrinceAndPauperQuantumParadoxHilbertSpaceFundamentalism} it was shown that the same {\MQS} can represent worlds corresponding to distinct outcomes of a quantum measurement. Such worlds are therefore physically distinct, and they can be transformed into one another by unitary operators that preserve both the Hamiltonian and the state vector.
This proves the abundance and rejects any hope that the only physically meaningful example of distinguishingness is time-distinguishingness.

\section{Applications to various candidate preferred structures}
\label{s:proof-TS-applications}

In this Section we will see that no preferred generalized basis, no preferred tensor product structure, no preferred emergent $3$D-space, not even a preferred quasi-classical level, can be defined from a {\MQS} $(\hilbert,\wh{H},\ket{\psi})$ alone.
To prove this, we reduce each of these structures to a Hermitian $\kind$-structure, and then we apply Theorem \ref{thm:nogo} to show for each case that physical relevance implies the existence of physically distinct possible choices.

\subsection{Non-uniqueness of the preferred basis}
\label{s:proof-PBS}

We will now prove the non-uniqueness of physically relevant generalized bases.

\begin{definition}
\label{def:POVM}
Let $(\hilbert,\wh{H},\ket{\psi})$ be a {\MQS}.
Let $\mc{A},\mc{B}'$ be sets, where $\mc{B}'$ may be the empty set, and let
\begin{equation}
\label{eq:POVM_B}
\mc{B}:=\mc{A}\cup\{0\}\cup\mc{B}'.
\end{equation}
A \emph{generalized basis} is a $\kind$-structure
\begin{equation}
\label{eq:POVM}
\struct_{\wh{H}}^{\ket{\psi}}=\(\wh{E}^{\ket{\psi}}_{\alpha}\)_{\alpha\in\mc{A}},
\end{equation}
with the following defining conditions.

The first condition is that for any $\alpha_0\in\mc{A}$ and any $\ket{\psi}\in\hilbert$, the operator $\wh{E}^{\ket{\psi}}_{\alpha_0}$ is \emph{positive semi-definite}, \ie
\begin{equation}
\label{eq:POVM_condition_positive_definite}
\mc{C}_{(\ket{\phi},\alpha_0)}\(\(\wh{E}^{\ket{\psi}}_{\alpha},\ket{\psi}\)_{\alpha}\):=\wh{E}^{\ket{\psi}}_{\alpha_0}\geq 0.
\end{equation}

The second condition is that, for any $\ket{\psi}\in\hilbert$, the operators $\(\wh{E}^{\ket{\psi}}_{\alpha}\)_{\alpha\in\mc{A}}$ form a \emph{resolution of the identity},
\begin{equation}
\label{eq:POVM_condition_resolution_I}
\mc{C}_{0}\(\(\wh{E}^{\ket{\psi}}_{\alpha}\)_{\alpha},\ket{\psi}\):=\wh{I}_{\hilbert}-\sum_{\alpha}\wh{E}^{\ket{\psi}}_{\alpha}=\wh{0}.
\end{equation}

We see from conditions \eqref{eq:POVM_condition_positive_definite} and \eqref{eq:POVM_condition_resolution_I} that $\struct_{\wh{H}}^{\ket{\psi}}$ is a \emph{positive operator-valued measure} (POVM).

The set $\mc{B}'$ is reserved for possible additional conditions
\begin{equation}
\label{eq:POVM_condition_psi}
\mc{C}_{\beta\in\mc{B}'}\(\(\wh{E}^{\ket{\psi}}_{\alpha},\ket{\psi}\)_{\alpha}\)=0,
\end{equation}
including possible conditions reflecting a dependence of $\wh{E}^{\ket{\psi}}_{\alpha}$ on $\ket{\psi}$.
\end{definition}

\begin{example}
\label{ex:POVM}
Particular cases of POVM are orthogonal bases (Example \ref{ex:hermitian_K_structure:basis}), \emph{projection-valued measures} (PVM) (projectors that give an orthogonal direct sum decomposition of the Hilbert space $\hilbert$), and overcomplete bases. All these cases are obtained by adding new conditions to the conditions \eqref{eq:POVM_condition_positive_definite} and \eqref{eq:POVM_condition_resolution_I} from Definition \ref{def:POVM}.

If we add the conditions that all $\wh{E}^{\ket{\psi}}_{\alpha}$ are projectors,
\begin{equation}
\label{eq:POVM_condition_proj}
\(\wh{E}^{\ket{\psi}}_{\alpha}\)^2-\wh{E}^{\ket{\psi}}_{\alpha}=\wh{0},
\end{equation}
and that all distinct $\wh{E}^{\ket{\psi}}_{\alpha}$ and $\wh{E}^{\ket{\psi}}_{\alpha'}$ are orthogonal,
\begin{equation}
\label{eq:POVM_condition_orthogonal}
\wh{E}^{\ket{\psi}}_{\alpha}\wh{E}^{\ket{\psi}}_{\alpha'}=\wh{0},
\end{equation}
we obtain a PVM that gives the orthogonal direct sum decomposition of $\hilbert$
\begin{equation}
\label{eq:direct_sum}
\hilbert=\boplus_\alpha\wh{E}^{\ket{\psi}}_{\alpha}\hilbert.
\end{equation}

If, in addition, we impose the condition that the projectors $\wh{E}^{\ket{\psi}}_{\alpha}$ are one-dimensional,
\begin{equation}
\label{eq:POVM_condition_vector}
\tr\wh{E}^{\ket{\psi}}_{\alpha}-1=0,
\end{equation}
we obtain the preferred basis from Example \ref{ex:hermitian_K_structure:basis}.
\end{example}

The following result applies to any of these kinds of generalized bases.

\begin{theorem}
\label{thm:POVM_nogo}
If there exists a physically relevant generalized basis of kind $\kind$, then there exist more physically distinct generalized bases of the same kind $\kind$.
\end{theorem}
\begin{proof}
For any of the $\kind$-structures from Definition \ref{def:POVM}, it follows immediately by applying Theorem \ref{thm:nogo}.
\end{proof}

\begin{remark}
\label{rem:basis_physical}
Preferred generalized bases are needed to define the macro states, or the branches of a branching structure in MWI, or decohered histories. In all these cases, the basis has to be physically relevant, because the components of the state vector in that basis are supposed to tell something about the physical state. For example, that the state changed in time, or that it could have been different. Without this, the candidate preferred basis is physically irrelevant. Therefore, Theorem \ref{thm:POVM_nogo} implies that such a basis cannot emerge uniquely from the {\MQS}.
\end{remark}

\begin{remark}
\label{rem:basis_total}
Here we considered the preferred generalized basis problem for the universe. The preferred generalized basis problem for subsystems, related to quantum measurements or selected by environmental decoherence, is a different issue, to be discussed in Sec. \sref{s:proof-PBS-sub}. However, the latter is a particular case of the former.
\end{remark}

\subsection{Non-uniqueness of the tensor product structure}
\label{s:proof-TPS}

The Hilbert space $\hilbert$ given as such, even in the presence of the Hamiltonian, does not exhibit a preferred tensor product structure.
Such a structure is needed to identify the subsystems, to address the preferred basis problem for subsystems, and to reconstruct the $3$D-space from the {\MQS}.
The Hamiltonian provides hints of what the tensor product structure may be by specifying what kinds of interactions and measurements are operationally accessible \cite{ZanardiLidarLloyd2004QuantumTensorProductStructuresAreObservableInduced}. But even so, and even if we impose additional conditions, the tensor product structure is not unique.

\begin{definition}
\label{def:TPS}
A \emph{tensor product structure} (TPS) of a Hilbert space $\hilbert$ is an equivalence class of unitary isomorphisms of the form
\begin{equation}
\label{eq:TPS}
\botimes_{\varepsilon\in\mc{E}}\hilbert_\varepsilon\mapsto\hilbert,
\end{equation}
where $\hilbert_\varepsilon$ are Hilbert spaces, and the equivalence relation is generated by local unitary transformations of each $\hilbert_\varepsilon$ and permutations of the set $\mc{E}$. The Hilbert spaces $\hilbert_\varepsilon$ represent subsystems, \eg they can be one-particle Hilbert spaces or algebras of local observables.
\end{definition}

It is known that there are multiple ways to choose the TPS, exhibiting different physical interactions \cite{DugicJeknicDugic2010WhichMultiverse,JeknicDugicEtAl2014QuantumStructuresOfAModelverseInconsistencyEverett}, which can even be reduced to simple phase changes \cite{Schwindt2012NothingHappensInEverettInterpretation}.
But this is to be expected, in the absence of defining conditions, whose role is to characterize the TPS precisely, uniquely if possible.

In order to apply Theorem \ref{thm:nogo} to prove that the TPS cannot be physically relevant and unique at the same time, we will show that the TPS is a $\kind$-structure.

\begin{proposition}
\label{thm:TPS_is_K-struct}
Any TPS is a $\kind$-structure.
\end{proposition}
\begin{proof}
Following \cite{ZanardiLidarLloyd2004QuantumTensorProductStructuresAreObservableInduced} and \cite{CotlerEtAl2019LocalityFromSpectrum}, we characterize the TPS in terms of operators on each factor space $\hilbert_\varepsilon$, extended to $\hilbert$. Let us define the subspaces of linear operators
\begin{equation}
\label{eq:Herm_factor}
\mc{L}_\varepsilon := \(\botimes_{\varepsilon'\neq\varepsilon}\wh{I}_{\varepsilon'}\)\otimes\mc{L}(\hilbert_\varepsilon),
\end{equation}
where the factor $\mc{L}(\hilbert_\varepsilon)$ is the algebra of linear operators on $\hilbert_\varepsilon$, and it is inserted in the appropriate position to respect a fixed order of $\mc{E}$.
For $\varepsilon\neq\varepsilon'\in\mc{E}$, if $\wh{A}_{\varepsilon}\in\mc{L}_{\varepsilon}$ and $\wh{B}_{\varepsilon'}\in\mc{L}_{\varepsilon'}$, then they commute.
Any operator from $\mc{L}(\hilbert)$ can be expressed as a real linear combination of products of operators from various $\mc{L}_\varepsilon$.

We will now make an extravagant choice for the set $\mc{A}$ needed to define the kind $\kind_{\tn{TPS}}$ for the TPS:
\begin{equation}
\label{eq:kind_TPS}
\mc{A}_{\tn{TPS}} := \bcup_{\varepsilon\in\mc{E}}\mc{L}_\varepsilon.
\end{equation}

We choose the tensors $\dwh{A}_\alpha$ giving our $\kind_{\tn{TPS}}$-structure as in Definition \ref{def:TS} to be the linear operators
\begin{equation}
\label{eq:tensor_TPS}
\wh{A}_{\alpha_\varepsilon} := \(\botimes_{\varepsilon'\neq\varepsilon}\wh{I}_{\varepsilon'}\)\otimes\wh{\alpha}_\varepsilon,
\end{equation}
where $\wh{\alpha}_\varepsilon\in\mc{L}(\hilbert_\varepsilon)$, 
and the defining conditions to be the commutativity of $\wh{A}_{\alpha_\varepsilon}$ and $\wh{A}_{{\alpha'}_{\varepsilon'}}$ for $\varepsilon\neq\varepsilon'$, and the completeness condition that the set of operators $\{\wh{A}_{\alpha_\varepsilon}|\alpha_\varepsilon\in\mc{A}_{\tn{TPS}}\}$ generates $\mc{L}(\hilbert)$.
This is equivalent to Definition \ref{def:TPS} \cite{ZanardiLidarLloyd2004QuantumTensorProductStructuresAreObservableInduced}.
Since we have characterized the TPS in terms of tensor structures, the TPS is a $\kind$-structure.
\end{proof}

\begin{theorem}
\label{thm:TPS_nogo}
If there exists a physically relevant TPS of a given kind $\kind$, then there exist more physically distinct TPS of the same kind $\kind$.
\end{theorem}
\begin{proof}
Proposition \ref{thm:TPS_is_K-struct} allows us to apply Theorem \ref{thm:nogo}. 
\end{proof}

\begin{remark}
\label{rem:TPS_physical}
A preferred TPS is needed to describe subsystems, for example particles. Therefore, it has to be physically relevant, and Theorem \ref{thm:TPS_nogo} implies that no preferred factorization can emerge uniquely from the {\MQS}.
\end{remark}

Proposition \ref{thm:TPS_is_K-struct} showed that any TPS is a $\kind$-structure, so that we could apply Theorem \ref{thm:nogo}. However, in practice, the set of operators \eqref{eq:tensor_TPS} is too extravagant, in the sense that the operators $\wh{A}_{\alpha_\varepsilon}$ corresponding to a fixed $\varepsilon\in\mc{E}$ can be transformed into one another by unitary transformations of $\hilbert_{\varepsilon}$. This would make it difficult to keep track of the indices $\alpha_{\varepsilon}$ when comparing the mean values $\bra{\psi}\wh{A}_{\alpha_\varepsilon}\ket{\psi}$ between unitary transformations of the Hilbert space $\hilbert$ to prove physical relevance.

Fortunately, it is sufficient to prove the physical relevance of $\kind_{\tn{TPS}}$-structures by using weaker invariants of the TPS. An example is given by the spectra of the reduced density operators $\wh{\rho}_{\varepsilon}$ obtained from the density operator $\wh{\rho}=\ket{\psi}\bra{\psi}$ by tracing over the space $\botimes_{\varepsilon'\neq\varepsilon}\hilbert_{\varepsilon'}$. They are sufficient to show the physical relevance of the $\kind_{\tn{TPS}}$-structures because there are many physically distinct states with $\wh{\rho}_{\varepsilon}$ having distinct spectra. 
The reason is that, in general, physical differences between states are reflected in physical differences of the subsystems, in particular in their corresponding reduced density matrices.
This shows that physically realistic TPS, by having to distinguish physically distinct states, are not unique.

More specific TPS, obtained by adding new defining conditions, will be used in \sref{s:proof-TPS-space}, where a model of emerging $3$D-space will be discussed. But whatever constraints we impose in order to obtain a unique TPS, there will be many TPS satisfying the same constraints, and even if it were unique, it would not be the TPS that we know from real physics (also see \cite{Stoica2024DoesTheHamiltonianDetermineTheTPSAndThe3dSpace}).

\subsection{Locality from the spectrum does not imply unique 3D-space}
\label{s:proof-TPS-space}

Now whatever candidate preferred TPS structure we may have in mind, we can express it as a $\kind$-structure by extending the $\kind_{\tn{TPS}}$-structure from Sec. \sref{s:proof-TPS} with new defining conditions.
In particular, anticipating our analysis of the attempt to reconstruct the $3$D-space from the {\MQS} described in \cite{CarrollSingh2019MadDogEverettianism,Carroll2021RealityAsAVectorInHilbertSpace}, we need to talk about the TPS reports from \cite{CotlerEtAl2019LocalityFromSpectrum}.

Cotler {\etal} obtained remarkable results concerning the TPS for which the interactions between subsystems are ``local'', in the sense that the interaction encoded in the Hamiltonian takes place only between a small number of subsystems \cite{CotlerEtAl2019LocalityFromSpectrum}. They showed that, in the rare cases when such a local TPS exists, it is \emph{almost} always unique up to an equivalence of TPS (\cf Definition \ref{def:TPS}). We do not contest their results, but we will see that, no matter how restrictive is the algorithm to obtain a local TPS from the spectrum of the Hamiltonian, it violates one of the Conditions \ref{cond:uniqueness} and \ref{cond:physical_relevance}, so either it is not unique, or it is not physically relevant. 
Since in reality the TPS is physically relevant, it follows that the local TPS is not unique.

Let us see what we need to add to the kind $\kind_{\tn{TPS}}$ from Sec. \sref{s:proof-TPS} to obtain the notion of local TPS from \cite{CotlerEtAl2019LocalityFromSpectrum}.
First, Cotler {\etal} expand the Hamiltonian $\wh{H}$ as a linear combination of products of operators $\wh{A}_{\alpha_\varepsilon}\in\mc{L}_\varepsilon$, defined in eq. \eqref{eq:Herm_factor}, such that each term is a product of operators $\wh{A}_{\alpha_\varepsilon}$ with distinct values for $\varepsilon\in\mc{E}$.
Then they impose the condition of locality for the TPS, which is that the TPS has to be such that the number of factors in each term of this expansion is not greater than some fixed small number $d\in\N$. This condition, which we will call \emph{$d$-locality}, is an inequation invariant to unitary transformations. Thus, we obtain a kind for the $d$-locality TPS, let us denote it by $\kind_{\tn{TPS-L(}d\tn{)}}$, where $\tn{L}$ stands for ``local'', and $d$ is the small number from the $d$-locality condition.

\begin{theorem}
\label{thm:TPSL_nogo}
If a $\kind_{\tn{TPS-L(}d\tn{)}}$-structure exists and is physically relevant, then there exist more physically distinct $\kind_{\tn{TPS-L(}d\tn{)}}$-structures.
\end{theorem}
\begin{proof}
The additional defining conditions required to make the TPS $d$-local can be expressed as tensor equations or inequations, as needed in the proof. These additional conditions make it more difficult to find a $\kind_{\tn{TPS-L(}d\tn{)}}$-structure, and indeed most Hamiltonians do not admit a local TPS \cite{CotlerEtAl2019LocalityFromSpectrum}. But when they admit one, either is not unique (which is allowed in \cite{CotlerEtAl2019LocalityFromSpectrum}), or, if it is unique, Theorem \ref{thm:nogo} implies that the $\kind_{\tn{TPS-L(}d\tn{)}}$-structure is unable to distinguish different states, and Condition \ref{cond:physical_relevance} is violated.
But in physically realistic situations the TPS is physically relevant, so in such situations either there is no $\kind_{\tn{TPS-L(}d\tn{)}}$-structure, or there are more physically distinct ones.
\end{proof}

Whenever a physically relevant $\kind_{\tn{TPS-L(}d\tn{)}}$-structure exists, more physically distinct ones exist. This does not challenge the results of Cotler {\etal}, but, as we will see, it shows that it cannot be used to recover a unique $3$D-space from the Hamiltonian's spectrum the way Carroll and Singh want \cite{CarrollSingh2019MadDogEverettianism,Carroll2021RealityAsAVectorInHilbertSpace}.

Carroll and Singh have an interesting idea to start from the local TPS and construct a space. For $d=2$, the $\kind_{\tn{TPS-L(2)}}$-structure defines a graph, whose vertices are labeled by elements of $\mc{E}$, and whose edges are the pairs $(\varepsilon,\varepsilon')$, $\varepsilon\neq\varepsilon'$, corresponding to the presence of an interaction between the subsystems $\hilbert_\varepsilon$ and $\hilbert_\varepsilon'$. Carroll and Singh interpret $\mc{E}$ as space, and the edges as defining its topology in a general sense. The topology of the space $\mc{E}$ depends only on the spectrum of the Hamiltonian. They also used the \emph{mutual information} between two regions $R,R'\subset\mc{E}$, $I(R:R')=S_R+S_{R'}-S_{RR'}$, where $S_R=-\tr\wh{\rho}_R\ln\wh{\rho}_R$ is the von Neumann entropy of $\ket{\psi}$ in the region $R$, to define distances between regions. They associate shorter distances to larger mutual information. Their program is to develop not only spacetime, but also emergent classicality, gravitation from entanglement {\textit{etc}.} \cite{CaoCarrollMichalakis2017SpaceFromHilbertSpace,CaoCarroll2018BulkEntanglementGravityWithoutABoundaryTowardsFindingEinsteinsEquationInHilbertSpace,CarrollSingh2019MadDogEverettianism,Carroll2021RealityAsAVectorInHilbertSpace}.
Their results are interesting, but unfortunately neither the resulting $3$D-space nor any other preferred structure in their program can emerge in a unique or essentially unique way from the Hamiltonian or its spectrum.

\begin{theorem}
\label{thm:TPS_space_nogo}
The emergent $3$D-space of Carroll and Singh \cite{CarrollSingh2019MadDogEverettianism,Carroll2021RealityAsAVectorInHilbertSpace} either is physically irrelevant, or it is not unique (not even essentially unique).
\end{theorem}
\begin{proof}
From Theorem \ref{thm:TPSL_nogo}, the structure defined in \cite{CarrollSingh2019MadDogEverettianism,Carroll2021RealityAsAVectorInHilbertSpace} as the topology of the $3$D-space is a $\kind_{\tn{TPS-L(}d\tn{)}}$-structure, and if it exists, it is not (essentially) unique.

Since these $3$D-space topologies are in various relations with $\ket{\psi}$, each of them is characterized by different values of the mutual information between their vertices, so the resulting distances will also not be unique.
\end{proof}

\begin{remark}
\label{rem:TPS_space_physical}
Any space has to be physically relevant, otherwise matter would look distributed the same in space at all times, and in alternative possible histories. Therefore, Theorem \ref{thm:TPS_space_nogo} implies that no $3$D-space can emerge uniquely from the {\MQS} by using the method from \cite{CarrollSingh2019MadDogEverettianism,Carroll2021RealityAsAVectorInHilbertSpace}.
\end{remark}

\begin{remark}
\label{rem:TPS_space_nogo}
In addition to Theorem \ref{thm:TPS_space_nogo}, there is another problem, 
which invalidates the method of Carroll and Singh, and any other method based on an underlying pre-metric $3$D-space given as a unique $\kind$-structure (including but not limited to the $\kind_{\tn{TPS-L(}d\tn{)}}$-structure resulting from the method of Cotler {\etal}).
No such construction that adds distances to a unique $\kind$-structure $\struct_{\wh{H}}$ can interpret $\struct_{\wh{H}}$ as a topological space, even in a general sense.
And this is true no matter what structure we add on top of it to obtain a metric, including the distances defined by mutual information. The reason is that if the position basis, or the sufficiently localized operators defining the small regions of the $3$D-space, are unique, then they will be preserved by the unitary time evolution, because they depend solely on the $\kind$-structure $\struct_{\wh{H}}$ assumed to be unique. But from the position-momentum Uncertainty Principle follows that the more localized is the quantum state at a time $t$, the more spread is immediately after $t$, due to the indeterminacy of the momentum. This also applies to localized operators. Therefore, such operators cannot have the right properties to be interpreted even as the topology (in any general sense) of the $3$D-space.

Another problem is that if the $\kind$-structure $\struct_{\wh{H}}$ assumed to be unique is interpreted as the topological structure of the $3$D-space, then the only source of physical relevance will come from the definition of distances. But during time evolution the relation between $\ket{\psi(t)}$ and $\wh{H}$ is preserved, so if $\struct_{\wh{H}}$ is unique, its relation with $\ket{\psi(t)}$ is also preserved, and such an approach cannot result in a realistic redistribution of matter in time, not even in local changes of the matter density due to the changes of the distances. For  example,  the electric charge  density at a point should be allowed to change its sign at different times. But even if $\struct_{\wh{H}}$ is not unique so that the distances are allowed to change in time, such a construction would only allow local changes of the volume element that is used to integrate the densities, but not changes of the sign of the charge density at the same point of the topological space.
But in realistic situations they should change as the system evolves.
\end{remark}


\subsection{Non-uniqueness of emergent 3D-space}
\label{s:proof-emergent-space}

\subsubsection{General emergent 3D-space}
\label{s:proof-emergent-space-general}

Now we will discuss generic kinds of \emph{emergent space} or \emph{emergent spacetime structure} (ESS) from a {\MQS} which may be a purely quantum theory of gravity.

We start with the NRQM case from Theorem \ref{thm:NRQM_nogo}, and generalize it to Quantum Field Theory (QFT).
The proof of Theorem \ref{thm:NRQM_nogo} can be reinterpreted in terms of the $\kind$-structures from the general Theorem \ref{thm:nogo} if we notice that the set $\mc{A}$ is the configuration space $\R^{3\n}$ and the $\kind$-structure is given by the  projectors
\begin{equation}
\label{eq:q}
\wh{A}_\q:=\ket{\q}\bra{\q},
\end{equation}
where $\q\in\mc{A}$. Then, the invariants can be chosen to be $(\bra{\psi}\wh{A}_\q\ket{\psi})_{\q\in\R^{3\n}}=(\abs{\braket{\q}{\psi}}^2)_{\q\in\R^{3\n}}$, which are invariant up to permutations of $\mc{A}$ corresponding to transformations of the configuration space $\R^{3\n}$. As seen in the proof of Theorem \ref{thm:NRQM_nogo}, transformations of the configuration space are not sufficient to undo the differences between distinct $\kind$-structures, and uniqueness is violated.

But this $\kind$-structure gives the configuration space, and we want one that gives the $3$D-space. So we rather choose $\mc{A}=\R^3$ and the operators
\begin{equation}
\label{eq:x}
\wh{A}_{\xThree}:=\ketbra{\xThree}{\xThree}=\sum_{j=1}^{\n}\intl{\R^{3(\n-1)}}\ket{\breve{\q}_j(\xThree)}\bra{\breve{\q}_j(\xThree)}\de\breve{\q}_j,
\end{equation}
where $\xThree\in\R^3$, $\q_j\in\R^3$ for $j\in\{1,\ldots,\n\}$ and
\begin{equation}
\label{eq:x_q}
\begin{aligned}
\ket{\breve{\q}_j(\xThree)}&:=\ket{\q_1,\ldots,\q_{j-1},\xThree,\q_{j+1},\ldots,\q_\n},\\
\de\breve{\q}_j&:=\de\q_1\ldots\de\q_{j-1}\de\q_{j+1}\ldots\de\q_\n.\\
\end{aligned}
\end{equation}

The Hermitian operators $\wh{A}_{\xThree}$ defined in eq. \eqref{eq:x} convey much less information than the operators $\wh{A}_\q$ from eq. \eqref{eq:q}, because they reduce the entire configuration space to the $3$D-space.
But the densities $\bra{\psi}\wh{A}_{\xThree}\ket{\psi}$ are still able to distinguish between infinitely many different values of $\ket{\psi}$, despite the symmetries of the $3$D-space, because the matter distribution in space allows to distinguish states.

We notice that we can deal with more types of particles by defining operators like in \eqref{eq:x} for each type, and we can also deal with superpositions of different numbers of particles, because now we are no longer restricted to the configuration space of a fixed number of particles. We can now move to the Fock space representation.

Let $\mc{P}$ be the set of all types of particles. We treat them as scalar particles, and move the degrees of freedom due the internal symmetries and the spin in $\mc{P}$.
The various components transform differently under space isometries (and more general Galilean or Poincar\'e symmetries) according to their spin, and also the gauge symmetries require them to transform according to the representation of the gauge groups, but these are encoded in the group $G_P$ from Definition \ref{def:equiv}.

We define $\mc{A}:=\mc{P}\times\R^3$. For each pair $(P,\xThree)\in\mc{A}$, let
\begin{equation}
\label{eq:N_P}
\wh{A}_{\(P,\xThree\)}:=\wh{N}_P(\xThree)=\wh{a}_P^\dagger(\xThree)\wh{a}_P(\xThree),
\end{equation}
where the operator $\wh{a}_P^\dagger(\xThree)$ creates a particle of type $P$ at the $3$D-point $\xThree\in\R^3$, $\wh{a}_P^\dagger(\xThree)\ket{0}=\ket{\xThree}_P$, $\ket{0}$ being the vacuum state, and $\wh{N}_P(\xThree)$ the \emph{particle number operator} at $\xThree$ for particles of type $P$.

This is all we need to represent the $3$D-space in QFT, since in the Fock space representation in QFT, everything is the same as in eq. \eqref{eq:N_P}, except that $\mc{P}$ represents now the types of fields instead of the types of particles, and we use in fact \emph{operator-valued distributions}.
We represent the states by state vectors from the Fock space obtained by acting with creation and annihilation operators on the vacuum state $\ket{0}$.

\begin{remark}
\label{rem:quantum-gravity-3dspace}
The $3$D-space may be much more general than the Euclidean space.

In \emph{quantum field theory on curved spacetime} \cite{Wald1994QFTCurvedSpacetime}, but also in \emph{Canonical Quantum Gravity} \cite{Dewitt1967QuantumTheoryOfGravityI_TheCanonicalTheory}, if we assume global hyperbolicity, as we usually do to have causality and a Hamiltonian formulation like \cite{ADM1962TheDynamicsOfGeneralRelativity} or \cite{Ashtekar1986NewVariablesForClassicalAndQuantumGravity}, the spacetime manifold has the form $M=\Sigma\times\R$, where $\Sigma$ is the $3$D-space manifold, and $\R$ represents time. The manifold $\Sigma$ may be $\R^3$, but it can be any $3$D-manifold.
So the operators defining the $3$D-space can be taken to satisfy eq. \eqref{eq:N_P} here too, but with $\mc{A}=\mc{P}\times\Sigma$.

In Canonical Quantum Gravity \cite{Dewitt1967QuantumTheoryOfGravityI_TheCanonicalTheory}, the {\schrod} wavefunctional formulation is used, so the canonical commutation relations are in terms of operators that represent the $3$D-metric and its canonical conjugate. But local operators still make sense, so we can still use $\mc{A}=\mc{P}\times\Sigma$ and eq. \eqref{eq:N_P}.

In various approaches to Quantum Gravity, $3$D-space is discrete.
Most, if not all such approaches, are based on structures that can be represented as graphs or hypergraphs, possibly having numbers attached to their vertices, edges, and/or hyperedges. Here are some examples. In the \emph{causal sets} approach \cite{Sorkin1990SpacetimeAndCausalSets}, causal sets are also graphs, whose vertices represent points of the spacetime manifold (\emph{events}), and the pairs of events in causal relation so that the first event is in the past lightcone of the second one are joined by edges.
In \emph{Regge calculus} \cite{Regge1961GeneralRelativityWithoutCoordinates}, spacetime is approximated by flat $4$-simplices. The information about the metric and curvature is concentrated at the edges and the $2$-faces \textit{etc}.
The \emph{causal dynamical triangulation} approach is similar, but all edges of the same length \cite{Loll2019QuantumGravityFromCausalDynamicalTriangulationsAReview}.
In \emph{Loop Quantum Gravity}, \emph{spin networks} are graphs having attached to the edges half-integer numbers corresponding to irreducible representations of the Lie algebra $\su(2)$ \cite{Penrose1971AngularMomentumAnApproachToCombinatorialSpaceTime,RovelliSmolin1995SpinNetworksAndQuantumGravity,AshtekarBianchi2021AShortReviewOfLoopQuantumGravity}. They represent $3$D-space, and in spacetime they become hypergraphs (\emph{spin foams}).

Since in these approaches the $3$D-space is replaced by the vertices of a graph or hypergraph, $\mc{A}=\mc{P}\times\Sigma$, where $\Sigma$ consists of the vertices of the graph or hypergraph and, if necessary, of its edges or hyperedges too.

Note that whether $\Sigma$ represents a $3$D-space manifold or a graph or hypergraph, most states in Quantum Gravity consist of linear combinations of states containing a definite geometry (or graph/hypergraph structure).
Therefore, in Quantum Gravity, one should expect that only some of the states have classical geometry. For example, in canonical Quantum Gravity on continuous manifolds like in \cite{Dewitt1967QuantumTheoryOfGravityI_TheCanonicalTheory}, the state has to be an eigenstate of the operator representing the $3$D-metric, $\wh{g}_{ab}(\xThree)$. In the discrete approaches, the graphs representing the $3$D-space also form a basis of the Hilbert space.

Since in these approaches to Quantum Gravity the states with classical geometry (continuous or discrete) are eigenstates of appropriate theory-dependent bases or operators, the recovery of the $3$D-space requires two steps. In the first step it is determined whether the state vector has a classical (continuous or discrete) geometry, by the condition that it is an element of the corresponding basis or an eigenvector of the corresponding operators (for example $\wh{g}_{ab}(\xThree)$). In the second step, the operators from eq. \eqref{eq:N_P} determine the $3$D-space structure.
\end{remark}

Since $\Sigma$ can even be a graph or hypergraph, as in the proposals mentioned in Remark \ref{rem:quantum-gravity-3dspace}, it can even consist of a finite set of vertices, so the Hilbert space can be finite-dimensional.
We will allow it to be finite-dimensional or infinite-dimensional, separable or not.
This is consistent with hints that in Quantum Gravity the Hilbert space may be finite or at least locally finite.

Operators \eqref{eq:N_P} commute with one another.
This is true for the bosonic case, but also for the fermionic case, because products of pairs of anticommuting operators commute with one another.
Therefore, the kind of structure that stands for ESS should be given in terms of commuting operators.
We can require them to represent the number of particles of type $P$ at a point of $\Sigma$, similar to the number operators $\wh{N}_P(\xThree)$ from \eqref{eq:N_P}.

Let us denote by $\kind_{ESS}$ the kind consisting of the conditions to be satisfied by a structure $\bigl(\wh{A}_{(P,\xThree)}\bigr)_{(P,\xThree)\in\mc{A}}$, $\mc{A}=\mc{P}\times\Sigma$, in order for it to be of the form \eqref{eq:N_P}.
We call this \emph{emergent space kind}, and denote it by  $\kind_{ESS}$. We call a $\kind_{ESS}$-structure $\bigl(\wh{A}_{(P,\xThree)}\bigr)_{(P,\xThree)\in\mc{A}}$ \emph{emergent space structure}.

\begin{theorem}
\label{thm:ESS_nogo}
If there is a $3$D-space structure of kind $\kind_{ESS}$, it has to be physically relevant, and then there are more possible physically distinct emergent $3$D-space structures of the same kind $\kind_{ESS}$.
\end{theorem}
\begin{proof}
The mean value $\bra{\psi}\wh{N}_P(\xThree)\ket{\psi}$ of the number of particles of each type at any point $\xThree$ in the $3$D-space structure should be allowed to change with the state, particularly as the state changes with time. Therefore, we expect $\bra{\psi}\wh{A}_{\(P,\xThree\)}\ket{\psi}$ to also depend on the state, in particular to change in time. Hence, the $\kind_{ESS}$ has to be physically relevant, and Theorem \ref{thm:nogo} implies that there are more possible physically distinct $3$D-space structures of the same kind $\kind_{ESS}$.
\end{proof}

Various defining conditions for the emergent $3$D-space may differ for different theories.
But whatever we will do, the unitary symmetry of the {\MQS} requires such a structure to satisfy Definition \ref{def:tensor}, and whatever defining conditions we impose to this structure, they have to be invariant to unitary symmetries, as in Definition \ref{def:kind}.
In all cases, Theorem \ref{thm:nogo} shows that any approach to recover space or spacetime as emerging from the {\MQS} alone leads to more physically distinct results.
This should not discourage such programs, only claims of uniqueness.

\subsubsection{Additional examples of emergent 3D-space}
\label{s:proof-emergent-space-additional}

\begin{example}
\label{ex:inclusion_maps_space_nogo}
An interesting idea to recover spacetime from a quantum theory was proposed by Giddings \cite{Giddings2019QuantumFirstGravity,DonnellyGiddins2017HowIsQuantumInformationLocalizedInGravity}. Following a profound analysis of how Local Quantum Field Theory extends to gravity, Giddings notices that the general relativistic gauge (diffeomorphism) invariance conflicts with usual tensor product decompositions, even in the weak field limit. For this reason, he rejects the idea of using a commuting set of observable subalgebras, and proposes instead a network of Hilbert subspaces $(\hilbert_\varepsilon)_\varepsilon\in\mc{E}$, where $\hilbert_\varepsilon\hookrightarrow\hilbert$ for all $\varepsilon\in\mc{E}$, and each $\hilbert_\varepsilon$ consists of state vectors in $\hilbert$ that are indistinguishable outside a neighborhood $U_\varepsilon$. For separated neighborhoods $U_\varepsilon$ and $U_{\varepsilon'}$, the condition $\hilbert_\varepsilon\otimes\hilbert_{\varepsilon'}\hookrightarrow\hilbert$ is also required. This approach is arguably more appropriate to define locality in the presence of gravity, and it defines a structure that is coarser than the usual spacetime. However, if we would want to start from the network of Hilbert subspaces $(\hilbert_\varepsilon)_\varepsilon\in\mc{E}$ and recover the spacetime structure or a coarse graining of it, non-uniqueness is unavoidable.
The reason is that such a network can be expressed in terms of projectors $(\wh{E}_\varepsilon)_{\varepsilon\in\mc{E}}$ on each of the subspaces in the network. Their incidence and inclusion relations, as well as the tensor product condition $\hilbert_\varepsilon\otimes\hilbert_{\varepsilon'}\hookrightarrow\hilbert$, are all invariant to unitary transformations, just like orthogonality is invariant in the case of an orthogonal decomposition of the Hilbert space $\hilbert$. Therefore, we can apply Theorem \ref{thm:nogo} just like in the case of decompositions into subspaces treated in \sref{s:proof-PBS}, and if we assume physical relevance, non-uniqueness follows.
\end{example}

\begin{example}
\label{ex:clock-ambiguity-space}
Albrecht and Iglesias propose that the clock ambiguity problem (also see Remark \ref{rem:clock-ambiguity}) can be used to make any Hamiltonian quasi-local, resulting in a notion of space \cite{AlbrechtIglesias2008ClockAmbiguityAndTheEmergenceOfPhysicalLaws,AlbrechtIglesias2012ClockAmbiguityImplicationsNewDevelopments}.
They argue that a random but time independent Hamiltonian should lead, in conjunction with the condition of quasi-separability of the physical law, to the emergence of space and physical laws.

In a similar spirit, Wetterich proposed earlier to construct distance, topology (assuming that the points of the spacetime manifold are given), and Riemannian metric from the correlation functions \cite{Wetterich1993GeometryFromGeneralStatistics}.
Later, Kempf also replaced spacetime distance by correlation functions \cite{Kempf2021ReplacingSpacetimeDistanceByCorrelation}.

If such approaches allow the recovery of the space or spacetime metric from the {\MQS} alone, without even knowing the points of the manifold, for example by operators as in \eqref{eq:N_P}, the result cannot be unique, as a result of Theorem \ref{thm:nogo}, and the resulting matter distribution cannot be realistic, \cf Remark \ref{rem:TPS_space_nogo}.
\end{example}

\begin{example}[Ergodic Hamiltonians]
\label{ex:ergodic}
Consider an \emph{ergodic Hamiltonian} $\wh{H}$, \ie a Hamiltonian whose eigenvalues are incommensurable \cite{Shnirelman1974ErgodicPropertiesOfEigenfunctions} (ergodic Hamiltonians were suggested to the author by Daniel Ranard). At least in the case when the Hilbert space is finite-dimensional, the one-parameter group of unitary evolution operators approximate, to any desired precision, all unitary operators that commute with $\wh{H}$. In other words, the one-parameter group $\bigl(\ee^{-\ii \wh{H}t}\bigr)_{t\in\R}$ is dense in the group of unitary operators that commute with $\wh{H}$. 

This means that for ergodic Hamiltonians there is effectively no other kind of distinguishingness but time-distinguishingness (Definition \ref{def:physical_relevance-time}).
It seems that, if $\wh{H}$ is ergodic, the results in this article, rather than being negative, show that the $4$D-spacetime uniquely emerges from the {\MQS}.
Maybe even that time itself emerges from a timeless $(\hilbert,\wh{H},\ket{\psi})$, as a continuous parametrization of all unitary operators that commute with $\wh{H}$.

A major problem with such a proposal is that the Hamiltonian has degenerate eigenvalues even for small systems, and in this case it cannot be ergodic.
Moreover, it appears that the spectrum of the Hamiltonian of our laws as we know them has continuous parts, due to the existence of particles that are not bound.

Both these properties of the Hamiltonian prevent the incommensurability of the eigenvalues. Hence, the possibility that the Hamiltonian is ergodic seems to be already excluded by observations.

Another problem is that, even if the set $\{\ee^{-\ii \wh{H}t}|t\in\R\}$ is dense in the set of unitary operators that commute with $\wh{H}$, a continuous mapping of a one-parameter group to a many-parameter group still does not exactly reduce all $3$D-space structures to a unique $4$D-spacetime, because the $3$D slices do not have a unique order relation that could be interpreted as time. 
\end{example}

\subsection{Non-uniqueness of branching into coherent states}
\label{s:proof-coherent_states}

A candidate preferred basis in NRQM is the position basis $\ket{\x}$, where $\x\in\R^{3\n}$ is a point in the classical configuration space.
But in NRQM there is another system of states, named \emph{coherent states}, that look classical and evolve approximately classically on short time intervals.
In the position basis $(\ket{\x})_{\x\in\R^{3\n}}$, squeezed coherent states $\ket{\p,\q}$ have the form
\begin{equation}
\label{eq:coherent_states}
\braket{\x}{\q,\p} := \(\frac{\ii}{\pi\hbar}\)^{\frac{3\n}{4}}\ee^{\frac{\ii}{\hbar}<\p,\x-\q/2>}\ee^{-\frac{1}{2\hbar}\abs{\x-\q}^2}
\end{equation}
for all points in the classical phase space $(\q,\p)\in\R^{6\n}$, where $<.,.>$ is the Euclidean scalar product in $\R^{3\n}$.

Coherent states were first used by {\schrod} \cite{Schrodinger1926DerStetigeUbergangVonDerMikroZurMakromechanik}, then by Klauder \cite{Klauder1960ActionOptionAndFeynmanQuantizationOfSpinorFields}, and in quantum optics by Sudarshan \cite{Sudarshan1963EquivalenceOfSemiclassicalAndQuantumMechanicalDescriptionsOfStatisticalLightBeams_CoherentStates} and Glauber \cite{Glauber1963CoherentAndIncoherentStatesOfTheRadiationField}. Coherent states form an overcomplete system, satisfying eq. \eqref{eq:POVM_condition_resolution_I}. By being highly peaked around points of the phase space, coherent states approximate well classical states, and their dynamics is close to the classical one for short time intervals. Therefore, they are good candidates for preferred generalized bases, and were indeed used as such to address the preferred (generalized) basis problem, \eg in \cite{JoosZeh1985EmergenceOfClassicalPropertiesThroughInteractionWithTheEnvironment,Omnes1988ConsistentHistoriesLogicalReformulationOfQuantumMechanicsIFoundations,Omnes1992ConsistentInterpretationsQuantumMechanics,Omnes1997QuantumClassicalCorrespondenceUsingProjectionOperators,Zurek1998DecoherenceEinselectionAndTheExistentialInterpretation,GellMannHartle1990DecoheringHistories,GellMannHartle1990QuantumMechanicsInTheLightOfQuantumCosmology,JoosZehKieferGiuliniKupschStamatescu2003DecoherenceAndTheAppearanceOfClassicalWorldInQuantumTheory,Wallace2012TheEmergentMultiverseQuantumTheoryEverettInterpretation}.

But can we recover the generalized basis of coherent states $\(\ket{\q,\p}\)_{(\q,\p)\in\R^{6\n}}$ from the {\MQS} alone?
This is equivalent to recovering the phase space by recovering the position and momentum operators.
Theorem \ref{thm:NRQM_nogo} already shows that it would not be unique, because its proof establishes the nonuniqueness of positions and momentum operators. Also Theorem \ref{thm:POVM_nogo} about generalized bases refutes this possibility.

An alternative definition of coherent states, which applies to QFT as well, is as eigenstates of the annihilation operators, see \eg \cite{Gazeau2009CoherentStatesInQuantumPhysics}.
But annihilation operators are characterized by certain defining conditions, like commutation or anticommutation relations they satisfy together with the creation operators, and by annihilating the vacuum state, which in its turn is defined as the state invariant to Poincar\'e transformations. These defining conditions have to be invariant.
And since they have to be physically relevant, Theorem \ref{thm:nogo} implies that they cannot be recovered uniquely from the {\MQS}.
Consequently, the coherent states as well cannot be recovered uniquely.

\subsection{Non-uniqueness of the preferred basis of a subsystem}
\label{s:proof-PBS-sub}

When one says that decoherence solves the preferred basis problem, this may mean two things. First, is that it leads to a preferred generalized basis of the entire world, and we have seen in \sref{s:proof-PBS} and \sref{s:proof-coherent_states} that this cannot happen if our only structure is the {\MQS}.
The second meaning one may have in mind is that of a preferred basis for a subsystem. This involves a factorization of the Hilbert space as a tensor product
\begin{equation}
\label{eq:factorization_system_environment}
\hilbert=\hilbert_S\otimes\hilbert_E,
\end{equation}
where $\hilbert_S$ is the Hilbert spaces of the subsystem and $\hilbert_E$ is that of the environment.
The environment acts like a thermal bath, and monitors the subsystem, making it to appear in an approximately classical state.
The reduced density operator of the subsystem, $\wh{\rho}_S:=\tr_E\ket{\psi}\bra{\psi}$, evolves under the repeated monitoring by the environment until it takes an approximately diagonal form.
Then we need an explanation of the fact that only one of the diagonal entries is experienced. The simplest such explanation is in MWI: all the diagonal entries of the reduced density matrix are real, and they correspond to distinct branches of the state vector.
This explanation is considered able to solve the measurement problem, since in this case the branches correspond to different outcomes of the measurement.
A toy model example was proposed by Zurek, who used a special Hamiltonian to illustrate decoherence for a spin $\frac 1 2$-particle, where the environment also consists of spin $\frac 1 2$-particles, \cf \cite{Zurek1982EnvironmentInducedSuperselectionRules} and \cite{Schlosshauer2007DecoherenceAndTheQuantumToClassicalTransition}, page 89.

It has been pointed out that such a decomposition is relative \cite{LombardiEtAl2012TheProblemOfIdentifyingSystemEnvironmentDecoherence}, to avoid the ``looming big'' problem of decoherence \cite{Zurek1998DecoherenceEinselectionAndTheExistentialInterpretation}. It has also been argued that such a decomposition may be observer-dependent \cite{Fields2014OnThePllivierPoulinZurekDefinitionOfObjectivity}.

This article does not contest the explanations based on decoherence. The question that interests us is whether the emergence of a preferred basis for the subsystem can happen when the only structure that we assume is the {\MQS}.
In particular, no preferred tensor product structure is assumed \emph{a priori}.

We have already seen that Theorem \ref{thm:TPS_nogo} implies that there are more physically distinct ways to choose a TPS, in particular a factorization like in eq. \eqref{eq:factorization_system_environment}. However, if we assume that the system and the environment are in separate states, \ie
\begin{equation}
\label{eq:factorization_system_environment_separable}
\ket{\psi}=\ket{\psi_S}\otimes\ket{\psi_E}
\end{equation}
at a time $t_0$ before decoherence leads to the diagonalization of $\wh{\rho}_S$,
then the possible ways to factorize the Hilbert space are limited. But is the factorization unique?

The first problem is that, even if we assume \eqref{eq:factorization_system_environment_separable}, in general there are infinitely many ways to choose the factorization \eqref{eq:factorization_system_environment}. Even for a system of two qubits there are infinitely many ways consistent with \eqref{eq:factorization_system_environment_separable}. We can choose a basis $(\ket{0},\ket{1},\ket{2},\ket{3})$ of the total Hilbert space $\hilbert\cong\C^4$ so that $\ket{0}=\ket{\psi}$, so $\ket{\psi}$ has the components $(1,0,0,0)$. We now interpret the basis as being obtained from a tensor product of the basis $(\ket{0}_S,\ket{1}_S)$ of $\hilbert_S\cong\C^2$ and the basis $(\ket{0}_E,\ket{1}_E)$ of $\hilbert_E\cong\C^2$, assuming that $\ket{0}=\ket{0}_S\otimes\ket{0}_E$, so $\ket{0}_S=\ket{\psi}_S$ and $\ket{0}_E=\ket{\psi}_E$. But, since there are infinitely many ways to construct the basis $(\ket{0},\ket{1},\ket{2},\ket{3})$, there are infinitely many ways to factorize $\hilbert$ as in eq. \eqref{eq:factorization_system_environment}, even up to local unitary transformations of the basis of each factor space. Even if we impose the restriction that the Hamiltonian has a particular form in the basis $(\ket{0},\ket{1},\ket{2},\ket{3})$, unless all of the eigenvalues of the Hamiltonian are distinct, there are infinitely many ways to choose it so that $\ket{0}=\ket{\psi}$. And in physically realistic theories the situation is even worse, since the Hamiltonian has highly degenerate eigenspaces, and in each of these eigenspaces it fails to impose any constraints on the factorization.
Adding more constraints cannot lead to a structure that avoids Theorem \ref{thm:nogo}.

An example of condition for this kind of decoherence is Zurek's \emph{commutativity criterion} \cite{zurek1981PointerBasisOfQuantumApparatus}, 
\begin{equation}
\label{eq:zurek-commutativity-criterion}
[\wh{H}_{\tn{int}},\wh{Z}_S\otimes\wh{I}_E]\approx 0,
\end{equation}
where $\wh{H}_{\tn{int}}$ is the interaction term of the Hamiltonian, and $\wh{Z}_S$ the pointer observable of the observed system. This condition establishes an invariant relation between the preferred decomposition \eqref{eq:factorization_system_environment}, the Hamiltonian, and the pointer observable.
The kind of such a candidate preferred structure contains a TPS of the form \eqref{eq:factorization_system_environment} and the pointer observable $\wh{Z}_S$, and they have to satisfy the TPS conditions from Proposition \ref{thm:TPS_is_K-struct} and condition \eqref{eq:zurek-commutativity-criterion}. 
And since both the TPS and the pointer observable have to be physically relevant, they satisfy Condition \ref{cond:physical_relevance}. By Theorem \ref{thm:nogo}, Condition \ref{cond:uniqueness} has to be violated.

Another way to see that a preferred basis for subsystems does not emerge uniquely from the {\MQS} is that this is in fact a particular case of the preferred basis for the total system. A basis of a subsystem $S$ can be defined in terms of a set of projectors for the total system, with the additional condition that the projectors commute with the algebra of local observables of the environment $E$.
But whatever conditions are added, they can be included in the additional conditions mentioned in Definition \ref{def:POVM}.
Therefore, the non-uniqueness of the preferred basis of a subsystem is a direct corollary of Theorem \ref{thm:POVM_nogo}.

But the major problem is that it would make no sense to assume that the systems are separated at the time $t_0$, precisely because of decoherence. The reason is that the subsystem of interest already interacted with the environment, and, unless we assume that it was projected, it is already entangled with the environment. Therefore, if we want to keep all the branches, as in MWI or consistent histories interpretations, we have to assume that in general $\ket{\psi(t_0)}$ is already entangled. And this prevents us from assuming that the state is separable. We would need first to have a preferred generalized basis, and we know from Theorem \ref{thm:POVM_nogo} that this is not possible from the {\MQS} alone. There are more physically distinct ways to choose the branching to start with, and then any solution for subsystems that is based on environmental-induced decoherence will depend on this choice.

\subsection{Non-uniqueness of classicality}
\label{s:proof-classicality}

From the previous examples we can conclude that the classical level of reality cannot emerge uniquely from the {\MQS} alone.

For this to be possible, we would need that the $3$D-space, and of course the factorization into subsystems like particles, emerge uniquely. But we have seen in \sref{s:proof-emergent-space} and \sref{s:proof-TPS} that this does not happen.

Another way for classicality to emerge would be if there was a preferred basis or generalized basis, which would correspond to states that are distinguishable at the macro level, but Theorem \ref{thm:POVM_nogo} prevents this too, as seen in \sref{s:proof-PBS}.

We can imagine that a clever combination of conditions, involving candidate preferred structures like a factorization into subsystems, locality in $3$D-space, and various relations between subsystems and their environments, may lead to an essentially unique solution.
But there is no limit of how many defining conditions we combine to define the kind of a preferred structure (Definition \ref{def:kind}). Too many conditions may leave us with no solution, and if a solution exists, many other physically distinct solutions exist, as shown by Theorem \ref{thm:nogo}.

\begin{example}
\label{ex:quantum-mereology}
A well-thought proposal was made by Carroll and Singh in \cite{CarrollSingh2021QuantumMereology}, based on a combination between factorization, locality, \emph{robustness} (the tendency of pointer states to remain unentangled with the environment), \emph{predictability} (the tendency of states nearby pointer states to remain peaked around classical trajectories in the appropriate regime) \textit{etc}. Their approach connects with topics discussed in \sref{s:proof-TPS}, \sref{s:proof-emergent-space}, and \sref{s:proof-PBS-sub}.

They start with eq. \eqref{eq:zurek-commutativity-criterion} as a condition for decoherence, and argue that decoherence is non-generic. As discussed in \sref{s:proof-PBS-sub}, the resulting structure is still far from being unique. Then they impose, on a system starting in a nonentangled pointer state, the condition that the entropy growth in time is minimized. Then they add the condition that the dynamics is predictable, in the sense that the density matrix becomes diagonal in the pointer basis, and in conjunction with the entropy growth condition, a quasi-classical evolution results.
The resulting conditions, expressed mathematically, lead to an algorithm that allows to obtain a quasi-classical factorization. They also study the emergence of classical conjugate positions and momenta-like variables. Since they work in a finite-dimensional Hilbert space, these are taken to be the \emph{generalized Pauli operators}, represented by \emph{clock} and \emph{shift} matrices, which satisfy approximate CCR relations \cite{Weyl1927QuantenmechanikUndGruppentheorie,Santhanam1976QuantumMechanicsInFiniteDimensions,Hall2013QuantumTheoryForMathematicians}, but in the infinite-dimensional limit the CCR \eqref{eq:CCR} is obtained. To recover such operators, they introduce the concept of \emph{operator collimation}.

The conditions associated to recovering the desired candidate preferred structures are reasonable and very restrictive.
These conditions being invariant, they define a kind which, because it aims to describe physically observable properties, satisfies Condition \ref{cond:physical_relevance}.
The resulting structures, when they exist, are not claimed to be unique in Carroll and Singh's article \cite{CarrollSingh2021QuantumMereology}, and they cannot be unique, because of Theorem \ref{thm:nogo}.
\end{example}

\section[``Paradoxes'': passive travel in time and in alternative\\\mbox{} realities]{``Paradoxes'': passive travel in time and in alternative realities}
\label{s:passive_travel}

One may hope that the non-uniqueness proved in this article is just a harmless feature of the Hilbert-space fundamentalism thesis. But it has bizarre consequences that may be problematic and should be understood.

\begin{problem}[Time machine problem]
\label{thm:time_machine}
If at the time $t_0$ there is a time-distinguishing $\kind$-structure $\struct_{\wh{H}}^{\ket{\psi(t_0)}}$, then at the same time there are more $\kind$-structures, with respect to which $\ket{\psi(t_0)}$ looks like $\ket{\psi(t)}$ for any other past or future time $t$.
This means that any system can ``passively'' travel in time by a simple unitary transformation of the preferred choice of the $3$D-space or of the preferred basis, so that the system's state looks with respect to the new ``preferred'' structure as if it is from a different time.

Therefore, Thesis {\HSF} implies that any state vector represents and supports not only the state of the system at the present time, but also the states of the system at any other past or future times, absolutely equally.
\end{problem}

\begin{problem}[Alternative realities problem]
\label{thm:alternative_reality}
According to Remark \ref{rem:distinguishing_special} and reference \cite{Stoica2023PrinceAndPauperQuantumParadoxHilbertSpaceFundamentalism}, there are forms of distinguishingness other than time-distinguishingness, for states not connected by unitary evolution and not distinguished by the Hamiltonian.
This means that there are more equally valid choices of the $3$D-space, of the TPS, or of the preferred basis, in which the system's state looks as if it is from an alternative world (and not in an Everettian sense).
Again, this leads to the in principle possibility of traveling in alternative realities, and also it means that the same state vector equally supports more physically distinct alternative realities, and there is no way to tell which one is the ``most real'' from the {\MQS} alone.
\end{problem}

\begin{problem}[Alternative laws problem]
\label{thm:alternative_laws}
Problems \ref{thm:time_machine} and \ref{thm:alternative_reality} assumed that only unitary transformations that preserve the form of the Hamiltonian are allowed. But is there any reason to impose this restriction? If not, if only the spectrum matters \cite{CotlerEtAl2019LocalityFromSpectrum,CarrollSingh2019MadDogEverettianism,Carroll2021RealityAsAVectorInHilbertSpace}, this would mean that passive travel in \emph{worlds having different evolution equations}, due to the Hamiltonian having a different form (albeit the same spectrum), is possible as well.
\end{problem}

\begin{remark}
\label{rem:clock-ambiguity}
A situation somewhat similar to Problem \ref{thm:alternative_laws} is related to the Page-Wootters solution to \emph{the problem of time in Quantum Gravity}, which consists of decomposing the vector $\kett{\Psi}$ from the Wheeler-DeWitt \emph{constraint equation}
\begin{equation}
\label{eq:Wheeler-DeWitt}
\wh{H}_{C+R}\kett{\Psi}=0
\end{equation}
in the form $\sum_{\tau\in\R}\ket{\tau}_C\ket{\psi(\tau)}_R$, where $\ket{\tau}_C$ represents a clock system encoding the time parameter, and $\ket{\psi(\tau)}_R$ the rest of the world \cite{PageWootters1983EvolutionWithoutEvolution}.

According to Albrecht and Iglesias, the Page-Wootters proposal suffers from the \emph{clock ambiguity problem}, \ie the possibility to change the evolution law into any other law by choosing another variable for the time parameter or the clock subsystem \cite{Albrecht1995TheoryOfEeverythingVsTheoryOfAnything,AlbrechtIglesias2008ClockAmbiguityAndTheEmergenceOfPhysicalLaws,AlbrechtIglesias2012ClockAmbiguityImplicationsNewDevelopments}.
In other words, given any two minimal quantum systems containing clocks and having isomorphic Hilbert spaces, they are equivalent under a change of the choice of the clock subsystem.

The Page-Wootters proposal is related to the third list of examples of {\MQS} from Remark \ref{rem:versatility} (more details in \cite{Stoica2022VersatilityOfTranslations}).

However, the clock ambiguity assumes implicitly an even stronger version of Thesis {\HSF}, since it assumes as the only given structures those from the Wheeler-DeWitt constraint equation \eqref{eq:Wheeler-DeWitt}, \ie the constraint Hamiltonian $\wh{H}_{C+R}$, rather than the dynamical Hamiltonian $\wh{H}=\wh{H}_R$, and the vector $\kett{\Psi}$, rather than $\ket{\psi}$.
These structures contain much less information even than the {\MQS}, since all this information reduces to the isomorphism class of the Hilbert space $\hilbert\otimes L^2\(\R\)$, as shown in \cite{Stoica2022VersatilityOfTranslations}.
\end{remark}

\begin{problem}[Alternative kinds]
\label{thm:alternative_kinds}
We discussed several different kinds of $3$D-space structures, in \sref{s:examples-NRQM-space}, \sref{s:proof-TPS-space} (Theorem
\ref{thm:TPSL_nogo}, Theorem \ref{thm:TPS_space_nogo}, Example \ref{ex:inclusion_maps_space_nogo}), and \sref{s:proof-emergent-space}. These different kinds of $3$D-space structures satisfy different defining conditions (Definition \ref{def:kind}).

What we call $3$D-space in a theory should correspond to what the empirical observations tell us about $3$D-space.
A correspondence should also exist between tensor product structures and what we identify empirically as subsystems.
Since empirically these structures are unique, the kind of $3$D-space structure and of TPS should also be unique, although it may be possible that different kinds cannot be distinguished by experiments, at least with the current technology.

Therefore, in addition to non-uniqueness of each particular kind, there are also multiple kinds that candidate for the role of $3$D-space, and similarly for the role of TPS.

Even in the same theory it is possible that different kinds of $3$D-space structure or different kinds of TPS emerge.
For example, the method to obtain a TPS from \cite{CotlerEtAl2019LocalityFromSpectrum} yields distinct kinds $\kind_{\tn{TPS-L(}d\tn{)}}$ for distinct values of $d$, and even for the same number $d$ it is possible that the resulting TPS are not unitarily equivalent.
This implies the same for the procedure of recovering the $3$D-space structure in \cite{CarrollSingh2019MadDogEverettianism,Carroll2021RealityAsAVectorInHilbertSpace}.
Another situation is encountered in \emph{M-theory}, where different superstring theories are conjectured to be related by dualities like the S- and T-duality, and to be obtained as limits of M-theory \cite{Duff1999MTheory}.
Also, when the AdS/CFT correspondence exists
\cite{Maldacena1999AdSCFTInitial}, to the same quantum theory on the boundary it does not necessarily correspond a unique bulk gravity theory, with a unique spacetime structure.

This implies that, if Thesis {\HSF} is true, the same physical system can support different theories of emergent $3$D-space or factorizations.
\end{problem}

\begin{remark}[The meaning of ``passive travel'']
\label{rem:passive-travel}
In Theorem \ref{thm:nogo}, to find out physically different structures of the same kind, we used \emph{active} transformations (see Remark \ref{rem:active-not-passive} and {\lessonName} \ref{lesson:invariance-fix}), but the ``passive travel'' from Problems \ref{thm:time_machine}--\ref{thm:alternative_kinds} has a complementary significance.

What $3$D-space and factorization appear to us as the preferred ones depends on how we experience the world through our sense organs and how our brains work.
But the configurations of our sense organs and brains depend on the $3$D-space and factorization.
Since the same state vector is consistent with multiple choices of the $3$D-space and factorization, it is consistent with multiple configurations of our organs and brains, each experiencing the world in its own $3$D-space and factorization. This applies to any $3$D-space and factorization allowed by the symmetry group of the {\MQS}.

If Thesis {\HSF} is correct, other $3$D-space structures and factorizations, with respect to which we may be in different configurations allowed by the symmetry group of the {\MQS}, are equally valid from the point of view of the {\MQS}.
They all exist without having to make an active transformation to construct them, because an active transformation can be compensated by a passive transformation. 

This raises the question, why is our experience consistent with a preferred $3$D-space and factorization, and not with all possible ones? It seems that the experienced worlds corresponding to different $3$D-space structures are independent, even if they are supported by the same state vector $\ket{\psi}$.

In some sense, this is similar to the notion of observer and reference frame, but there is an important difference. In Special Relativity, the same system can be described in any frame, but not any frame is the reference frame of an observer.
Different observers in different frames can be aware of one another, and can objectively observe one another.
By contrast, Problems \ref{thm:time_machine}--\ref{thm:alternative_kinds} reveal that, if Thesis {\HSF} is true, the present state of the universe is populated with all the observers that existed or will exist at any time, and in different possible alternative realities. They do not exist in different reference frames, but in different $3$D-space structures. However, the way $\ket{\psi}$ looks in two different $3$D-space structures can be related by a passive change of basis of the Hilbert space, similar to the change of frame in Special Relativity, but resulting in different worlds.

Another similarity can be made to the question of experiencing multiple worlds in MWI. The difference is that Thesis {\HSF} implies multiplicity of experienced worlds for the same state vector, while MWI implies multiplicity because the state vector has different components, corresponding to different branches of the wavefunction.
\end{remark}

The existence of these problems does not necessarily mean that one can actually travel in time and in ``alternative worlds'' with the same or with different laws like this in practice, but it at least means that, at any time, there is a sense in which all past and future states, as well as ``alternative worlds'' which are not due to any version of the Many-Worlds Interpretation, are ``simultaneous'' with the present state, being supported by the same state.

Moreover, in the case of MWI and other branching-based interpretations, for every branching structure there are more alternative branching structures.
The proliferation of such ``basis-dependent worlds'' is ensured by the physical relevance of the candidate preferred $3$D-space, TPS, or generalized basis, under the assumption that the only fundamental structure is the {\MQS}, and everything else should be determined by this.

\begin{remark}
\label{rem:so-what}
Theorem \ref{thm:nogo} shows that such structures cannot emerge uniquely from the {\MQS} alone, but \emph{it does not say that they have to be unique}.

The fact that relying on the {\MQS} alone entails Problems \ref{thm:time_machine}--\ref{thm:alternative_kinds} and maybe other problems is not necessarily a ground to refute the possibility.
Maybe this is how the universe is, with multiple alternative $3$D-spaces and factorizations into subsystems.
But would we even be able to know it, given that we are confined to our own perspective, which depends on a particular choice of the $3$D-space structure and of the factorization into subsystems?

This article does not claim to offer an answer, all it does is to refute the uniqueness claim from Thesis {\HSF}.
\end{remark}

These problems cannot simply be dismissed, they should be investigated to see if indeed the observers in a basis-dependent world cannot change their perspective to access information from other basis-dependent worlds allowed by unitary symmetry.

In Sec. \sref{s:avoid-problem-MWI-wavefunction}, we will see that it is not clear how to avoid these problems even if we extend the {\MQS} with additional structures for the $3$D-space, the factorization into subsystems \textit{etc}., as in the ``wavefunction version''.

\section[What approaches to Quantum Mechanics avoid the\\\mbox{} problems?]{What approaches to Quantum Mechanics avoid the problems?}
\label{s:avoid-problem}

How should we resolve these problems for theories like the Hilbert space fundamentalist versions of Standard Quantum Mechanics or of Everett's interpretation? 
Some implications and available options of the too symmetric structure of the Hilbert space were already discussed in the literature, see \eg \cite{Schmelzer2011PureQuantumInterpretationsAreNotViable,Schwindt2012NothingHappensInEverettInterpretation}.

Let us explore several possible ways to break the unitary symmetry so that it allows the unique emergence of $3$D-space and of factorization into subsystems.

\subsection{Bohr and Heisenberg}
\label{s:avoid-problem-CI}

Bohr's interpretation, in which the measuring apparatus is a classical system and the observed system is quantum, are protected from the consequences of Theorem \ref{thm:nogo}, precisely because the measuring device is considered to have classical properties. In particular, its components have definite classical positions, so the problem of the preferred $3$D-space does not occur for the measuring device. Moreover, by interacting with the observed particle and finding it in a certain place, this knowledge of the $3$D-space is extended from the system of the measuring device to that of the observed particle. But there is of course a price: the theory does not include a quantum description of the measuring device, hence it is not universal. The problem of recovering the $3$D-space is avoided by ``contaminating'' the quantum representation of the observed system with information that can only be known from the macro classical level.
Heisenberg's view also does not provide unified laws for both the quantum and the classical regimes, particularly the measuring devices.

However, it is possible to formulate QM in a way that gives a unified description of the classical and quantum regimes and introduces the necessary symmetry breaking, see \eg Stoica \cite{Stoica2020StandardQuantumMechanicsWithoutObservers}.

\subsection{Embracing the symmetry}
\label{s:avoid-problem-anything}

A possible response to the implications of Theorem \ref{thm:nogo} may be to simply bite the bullet and embrace its consequences. For example, in MWI, we can pick a preferred space and a preferred TPS, but accept that there are more possible ways to do this, as suggested by Saunders \cite{Saunders1995TimeQuantumMechanicsAndDecoherence}. Everything will remain the same as in MWI, but there will be ``parallel many-worlds''. Each of the ``$3$D-space-dependent many-worlds'' are on equal footing with the others. In one-world approaches with state vector reduction we will also have multiple $3$D-space-dependent worlds, but different from the many-worlds, in the sense that they are not branches, they are just one and the same world viewed from a different $3$D-space. But if we accept all of the possible $3$D-space-dependent worlds allowed by symmetry, we should also try to show that this position does not have unintended consequences. In particular, how can one prove that passive travel in time and in parallel worlds (\cf Sec. \sref{s:passive_travel}) is not possible?
How does the existence of multiple worlds for the same state-vector affects the probabilities? This seems to entail deviations from the Born rule (\cite{Stoica2023TheRelationWavefunction3DSpaceMWILocalBeablesProbabilities}, Appendix C).

Maybe there is a reason to keep only some of the $3$D-space-dependent worlds.
To satisfy the Second Law of Thermodynamics, the universe has to satisfy the \emph{past hypothesis} \cite{BarryLoewer2020MentaculusVision}.
Similar constraints are required for the existence of decohering branching structures that only branch into the future and not into the past \cite{Wallace2012TheEmergentMultiverseQuantumTheoryEverettInterpretation}.
Moreover, even the simple requirement of having the records consistent with the past history of the system requires very special initial conditions, bordering conspiracy \cite{Stoica2022DoesQuantumMechanicsRequireConspiracy}.
All these problems require the initial state of the universe to be fine-tuned in a very special way.

But if the initial state itself depends on the choice of the candidate preferred structures, then most such choices would fail to ensure special enough initial conditions. Hence, we may have a way to reduce the number of possible alternative candidate preferred structures: choose only those $3$D-space structures and factorizations consistent with the Second Law of Thermodynamics, containing observers able to remember the past and not the future, and, in the case of MWI, ensuring that decoherence takes place properly to ensure branching towards the future only.

However, this is not restrictive enough to ensure uniqueness: if the Past Hypothesis, the condition of having valid records, and the conditions needed for the branching structure depend on the {\MQS} alone, they involve constraints that are invariant under unitary transformations, so there will still be more possible choices.
Moreover, if we have to rely on systems that count as observers to select the $3$D-space-dependent world satisfying these constraints, this may fail, for the simple reason that most such observers may be \emph{Boltzmann brains} \cite{schulman1997timeArrowsAndQuantumMeasurement,AlbrechtSorbo2004BoltzmannBrain} (also see \cite{Stoica2023AreObserversReducibleToStructures} and Sec. \sref{s:avoid-problem-MWI-wavefunction} for more about observers).

\subsection{Enforcing a preferred structure}
\label{s:avoid-problem-MWI-wavefunction}

Instead of using the abstract state vector, can we simply use the wavefunction, or simply extend the {\MQS} structure with a preferred TPS and a preferred $3$D-space structure?
Even if we extend the {\MQS} with an additional structure, one may object that the unitary symmetry makes any such structure  indistinguishable from any other one obtained by unitary symmetry from it, making the theory unable to predict the empirical observations, which clearly emphasize particular observables as representing positions, momenta \textit{etc}.

But such an objection would only be fair if we would apply the same standard to Classical Mechanics. In Classical Mechanics, we can make canonical transformations of the phase space that lead to similar problems like those discussed here for the Hilbert space. Yet, this is not seen as a problem, because we take the theory as representing real things, like particles and fields, propagating in space.
We may therefore think that we can just do the same and adopt the more common claim of MWI that the state vector is in fact a wavefunction.

To ground the wavefunction even more in the $3$D-space, we can even represent it as multiple classical fields on the $3$D-space, as shown by Stoica \cite{Stoica2019RepresentationOfWavefunctionOn3D,Stoica2022BackgroundFreedomLeadsToManyWorldsLocalBeablesProbabilities}. The representations from  \cite{Stoica2019RepresentationOfWavefunctionOn3D,Stoica2022BackgroundFreedomLeadsToManyWorldsLocalBeablesProbabilities} are fully equivalent to the wavefunction on the configuration space, work for all interpretations, and can provide them the necessary structure and a $3$D-space ontology.

Another possibility is Barbour's proposal, in which the worlds are configurations from a certain configuration space, specific to Barbour's approach to Quantum Gravity based on \emph{time capsules} and \emph{shape dynamics}, but which can easily be generalized into a general version of Everett's Interpretation. In Barbour's proposal, the wavefunction is interpreted as encoding the probabilities of these configurations \cite{Barbour2001TheEndOfTime}.

But the things are not that simple with Everett's Interpretation and MWI in general. The reason is that, in order to give an ``ontology'' to the branches, the solution is, at least for the moment, to interpret the physical objects as patterns in the wavefunction, and apply \emph{Dennett's criterion} that ``patterns are real things'' as Wallace calls it (\cite{Wallace2003EverettAndStructure} p. 93 and \cite{Wallace2012TheEmergentMultiverseQuantumTheoryEverettInterpretation} p. 50). For Dennett's notion of pattern see  \cite{Dennett1991RealPatterns}. The key idea can be stated as ``a simulation of a real pattern is an equally real pattern''. For a criticism by Maudlin of the usefulness of this criterion as applied by Wallace see \cite{Maudlin2014CriticalStudyDavidWallaceTheEmergentMultiverseQuantumTheoryEverettInterpretation} p. 798. 

We will have more to say about this in another paper, in the context of the results presented here \cite{Stoica2023AreObserversReducibleToStructures}. We will see that the solution depends on the ability of the preferred structure to guarantee experiences of the world in a way unavailable to other unitary transformations of that structure. And this depends on the theory of mind, since for example a computationalist theory of mind allows the transformed (``simulated'') patterns obtained by unitary transformations of ``real'' patterns to have the same experiences as the ``real'' ones, because whatever computation performed on the ``real'' patterns is supposed to yield conscious experiences, it is identical to a computation on the transformed patterns. Therefore, since at least Wallace's approach based on Dennett's idea of pattern, and in fact Everett's original idea and subsequent variations \cite{Tegmark2015ConsciousnessAsAStateOfMatter,Zurek1998DecoherenceEinselectionAndTheExistentialInterpretation,Zurek2009QuantumDarwinism}, are implicitly or explicitly committed to computational or functionalist theories of mind, enforcing a preferred structure leads us in the same place as the option of embracing the symmetry mentioned in \sref{s:avoid-problem-anything}. 

To have a single generalized basis or a single underlying $3$D-space, one may need to assume that not all patterns that look like mental activity of observers actually support consciousness, and the preferred structures are correlated to those that support it.
Maybe this requires a version of the von Neumann-Wigner Interpretation of QM, or a version of the \emph{Many Minds} Interpretation \cite{Zeh1970OnTheInterpretationOfMeasurementInQuantumTheory,AlbertLower1988InterpretingMWI_ManyMinds,Albert92QuantumMechanicsAndExperience} where consciousness is not completely reducible to a pattern as in the computationalist or functionalist theories of mind \cite{Stoica2023TheRelationWavefunction3DSpaceMWILocalBeablesProbabilities}.

\subsection{Breaking the unitary symmetry of the laws}
\label{s:avoid-problem-break-laws}

Maybe the problem is better solved in theories that actively break the unitary symmetry of the very laws of QM, by either modifying the dynamics or including objects like particles living in the $3$D-space.

An example is the \emph{Pilot-Wave Theory} \cite{deBroglie1956TheorieDoubleSolution,Bohm1952SuggestedInterpretationOfQuantumMechanicsInTermsOfHiddenVariables,BohmHiley1993UndividedUniverse,CushingFineGoldstein1996BohmianMechanicsAppraisal,DGZ1996BohmianMechanics} and variants like \cite{Nelson2020DynamicalTheoriesOfBrownianMotion}, which extend the {\MQS} with a preferred $3$D-space and point-particles with definite positions in the $3$D-space, and which breaks the unitary symmetry of the Hilbert space.

\emph{Objective Collapse Theories} \cite{GhirardiRiminiWeber1986GRWInterpretation} and variations like \cite{Pearle1989StochasticSpontaneousLocalization,Diosi1987GravitationalQuantumCollapse,Penrose1996OnGravitysRoleInQuantumStateReduction} also supplement the Hilbert space with $3$D positions, and the wavefunction collapses spontaneously into a highly peaked wavefunction well localized in the configuration space. Even if the $3$D-space is not assumed to be fixed, it may be possible to recover the configuration space one-collapse-at-a-time, and then the $3$D-space may emerge from the way interactions depend on the distance (see Remark \ref{rem:3dspace:nrqm:specs} and \cite{Albert1996ElementaryQuantumMetaphysics}).

\subsection{Breaking the unitary symmetry of the space of solutions}
\label{s:avoid-problem-break-solutions}

Could the needed symmetry breaking of the group of unitary transformations be due to restricting the physically allowed solutions, rather than on changing the laws? This would be a more conservative way to solve the problem, since it would not require to modify or supplement the {\schrod} dynamics, while keeping a single world. 

There are already known proposals that not all state vectors in the Hilbert space describe real physics, in approaches that try to maintain unitary evolution but select the physically allowed states so that the appearance of the state vector reduction is done without having to appeal to branching into many worlds or to state vector reduction. Such proposals were made by Schulman \cite{Schulman1984DefiniteMeasurementsAndDeterministicQuantumEvolution,schulman1997timeArrowsAndQuantumMeasurement,Schulman2017ProgramSpecialState}, 't Hooft \cite{tHooft2016CellularAutomatonInterpretationQM}, and Stoica \cite{Stoica2012QMGlobalAndLocalAspectsOfCausalityInQuantumMechanics,Stoica2017TheUniverseRemembersNoWavefunctionCollapse,Stoica2021PostDeterminedBlockUniverse}. 

Admittedly, these proposals seem ``conspiratorial'' or ``retrocausal'' as per Bell's Theorem \cite{Bell64BellTheorem,Bell2004SpeakableUnspeakable}, but the gain is relativistic locality and restoration of the conservation laws and of the {\schrod} dynamics, for a single world rather than the totality of worlds as in MWI.
And there are ways to interpret this less dramatically than as conspiratorial \cite{Stoica2021PostDeterminedBlockUniverse}.

Neither of these proposals assumes the {\MQS} as the only structure, in general they assume a space structure too. In this sense, they avoid the results in this article, although this may not be enough, as suggested in \sref{s:avoid-problem-MWI-wavefunction} and \cite{Stoica2023AreObserversReducibleToStructures}. 

However, we may ask if, by accepting as representing physical states only a subset $\hilbert_0\subset\hilbert$ of the Hilbert space $\hilbert$, where $\ee^{-\ii \wh{H}t}\hilbert_0=\hilbert_0$ for any $t\in\R$, the unitary group is sufficiently reduced to allow the recovery of space from the {\MQS} and $\hilbert_0$. The answer seems to be negative, since at least the one-parameter group generated by the Hamiltonian will not be broken, and therefore at least the consequences of Corollary \ref{thm:nogo_time}, and at least Problem \ref{thm:time_machine}, cannot be avoided in this way either.

But maybe the restriction to $\hilbert_0$ ensures these structures to uniquely emerge up to the time parameter, in which case, while the $3$D-space cannot be recovered, spacetime as a whole can be recovered, as suggested in Example \ref{ex:ergodic}. It may be the case that allowing only a subset $\hilbert_0$ of the Hilbert space can ensure this without requiring the Hamiltonian to be ergodic as in Example \ref{ex:ergodic}. 
An extreme example is $\hilbert_0=\{\ee^{-\ii \wh{H}t}\ket{\psi_0}|t\in\R\}$ for some $\ket{\psi_0}\in\hilbert$. This indeed gives a unique spacetime, given an emergent notion of $3$D-space. But this example is artificially constructed \emph{ad-hoc}, a genuine answer should of course be derived from some physical principles.
This possibility remains to be explored.

\bigskip
\noindent\textbf{Acknowledgement.}
The author thanks Louis Marchildon, Tim Maudlin, Heinrich P\"as, Daniel Ranard, Igor Salom, Jochen Szangolies, Iulian Toader, Lev Vaidman, Howard Wiseman, and the referees, for feedback and stimulating exchanges.


\address{Deptartment of Theoretical Physics, NIPNE---HH\\
M{\u a}gurele, Romania 077125\\
\email{cristi.stoica@theory.nipne.ro, holotronix@gmail.com}}

\end{document}